\documentclass[sigconf]{acmart}

\renewcommand\footnotetextcopyrightpermission[1]{} 
\usepackage{amsthm}
\usepackage{algorithmic}
\usepackage[ruled,linesnumbered]{algorithm2e}
\usepackage{graphicx}
\usepackage{subcaption}
\usepackage{epstopdf}
\usepackage{tabularx}  
\usepackage{fancyhdr}
\usepackage{comment}

\newtheorem{assumption}{Assumption}

\newtheorem{theorem}{Theorem}

\AtBeginDocument{%
  \providecommand\BibTeX{{%
    Bib\TeX}}}

\setcopyright{acmlicensed}
\copyrightyear{2018}
\acmYear{2018}
\acmDOI{XXXXXXX.XXXXXXX}

\acmConference[MobiHoc 2024]{Make sure to enter the correct
  conference title from your rights confirmation emai}{October 14--17,
  2024}{Athens, Greece}
\acmISBN{978-1-4503-XXXX-X/18/06}

\begin{document}

\title{Resource Allocation and Secure Wireless Communication in the Large Model-based Mobile Edge Computing System}

\author{Zefan Wang}
\affiliation{%
  \institution{Nanyang Technological University}
  \country{Singapore}
}
\email{zefan001@e.ntu.edu.sg}

\author{Yitong Wang}
\affiliation{%
  \institution{Nanyang Technological University}
  \country{Singapore}}
\email{yitong002@e.ntu.edu.sg}

\author{Jun Zhao}
\affiliation{%
  \institution{Nanyang Technological University}
  \country{Singapore}}
\email{junzhao@ntu.edu.sg}


\begin{abstract}
  With the rapid advancement of large models and mobile edge computing, transfer learning, particularly through fine-tuning, has become crucial for adapting models to downstream tasks. Traditionally, this requires users to share their data with model owners for fine-tuning, which is not only costly but also raises significant privacy concerns. Furthermore, fine-tuning large-scale models is computationally intensive and often impractical for many users. To tackle these challenges, we introduce a system that combines offsite-tuning with physical-layer security, which provides local data owners with a lightweight adapter and a compressed emulator. Data owners then fine-tune the adapter locally and securely send it back to the model owners through a confidential channel for integration, ensuring privacy and resource conservation. Our paper focuses on optimizing computational resource allocation among data owners and the large model owner deployed on edge, and on the compression ratio of adapters. We incorporate a secrecy uplink channel to maximize the utility that we defined while minimizing system costs like energy consumption and delay. The optimization uses the Dinkelbach algorithm, fractional programming, successive convex approximation and alternating optimization. Experiments demonstrate our algorithm's superiority over existing methods.
\end{abstract}




\maketitle
\thispagestyle{plain}
\pagestyle{plain}

\section{Introduction}
Large models (LM)\footnote{Although many studies simply say ``foundation models'', we adopt the name ``large models'' to highlight both the sheer size and the foundation nature of the models. Such a name is also used in prior work~\cite{rawte2023survey}.} are generative machine learning models pre-trained on large-scale unlabeled datasets (e.g., BookCorpus or English Wikipedia~\cite{bommasani2021opportunities}). Pre-trained models have the capability to acquire generalized representations that can be effectively utilized in fine-tuning by using smaller, task-specific datasets supplied by various downstream users to cater to a diverse array of downstream tasks~\cite{hu2021lora}, where models like GPT-4 by OpenAI stand out with their sophisticated understanding and creation of human-like text~\cite{devlin2018bert}. Meanwhile, LMs like DALL-E and CLIP demonstrate remarkable abilities in understanding and generating visual content~\cite{radford2021learning}, and Google-USM achieving new levels of accuracy and naturalness in speech-to-text conversion~\cite{zhang2023google}.

\textbf{Challenges and motivation.} 
There are several challenges in the process of fine-tuning LMs. On the one hand, there is a desire to exploit the full capabilities of LMs, on the other, the sensitive information has potential risks of exposure to LM service providers. This risk is non-neglectable in many applications such as healthcare or finance~\cite{behnia2022ew, xu2023federated}. A potential solution could be localized deployment, where the large model owner (LMO) deploys the model on the client. However, due to the large structures of the LMs, deployment and inference locally are essentially impossible, for instance, GPT-4 consists of an astonishing 175 billion parameters~\cite{achiam2023gpt}. Moreover, local deployment risks may leak the models' intellectual property (IP). Offsite-tuning has been proposed as a viable Parameter Efficient Fine-Tuning (PEFT) solution~\cite{ding2023parameter, xiao2023offsite}. In this approach, the LMO sends a lightweight adapter and a lossy compressed emulator to the data owner (DO). The DO then fine-tunes the adapter using the emulator. Once fine-tuned, the adapter is returned to the LMO and integrated into the full model to create an adapted FM. Under offsite-tuning, the DO does not need to share their training data, and the LMO does not share the full model weights.

A significant security challenge in the context of offsite-tuning is the potential eavesdropping on the communication channel, where eavesdroppers could intercept the adapter during transmission. If successful, they might use the data within the adapter to infer sensitive information about the DO's local data or objectives. Eavesdroppers can passively gather intel without altering the data, such as in VoIP eavesdropping, or actively insert themselves into the network, manipulating the data as in man-in-the-middle attacks. These attacks can lead to the theft of intellectual property, passwords, and other sensitive data, significantly compromising the security and privacy of the data owner~\cite{ma2018security, kapetanovic2015physical}. In response to these challenges, our primary motivation is to develop secure wireless communication in an LM-based edge computing system, which cooperates with the PEFT technique for resource-saving.  

\textbf{Contributions.} The major contributions of our paper are summarized as follows:

\begin{itemize}
    \item We formulate the problem of maximizing the users’ utility-consumption ratio (UCR) under physical-layer security for the offsite-tuning process of the LM. To our knowledge, we first study and investigate UCR in the context of the increasingly critical and emergent domain of offsite-tuning for LMs.
    \item We develop an advanced optimization algorithm by integrating the Dinkelbach algorithm, fractional programming, Successive Convex Approximation (SCA), and alternating optimization techniques. This innovative approach enables our algorithm to efficiently address and solve complex UCR optimization challenges, establishing a new benchmark for effectiveness in this field.
    \item We demonstrate through simulations that our algorithm significantly outperforms existing baselines, adeptly coordinating the optimization of adapter parameter offloading for users, resource allocation, and the balance between user utility and cost. This evidence showcases our algorithm's superior capability to navigate and optimize the multifaceted aspects of the offsite-tuning process, thereby advancing both operational efficiency and security.
\end{itemize}


\textbf{Roadmap.} The organization of this paper is as follows. Section~\ref{sec_rela} presents related work. Section~\ref{sec_sys} introduces the proposed system model and the formulated problem. Our proposed solution, encompassing the algorithmic approach, is meticulously expounded in Section~\ref{sec_solu}. Section~\ref{sec_exp} provides the numerical results, and Section~\ref{sec_con} summarizes our paper.

\section{Related Work} \label{sec_rela}
This section explores relevant research across three focused areas pertinent to our study. Section~\ref{rela_1} is mainly devoted to current research in solving resource allocation problems over wireless communication. Section~\ref{rela_2} focuses on integrating large models (LMs) with edge computing, and Section~\ref{rela_3} focuses on the physical layer security in mobile communication.

\subsection{Resource allocation over wireless communication} \label{rela_1}
Recently, several studies have successfully addressed resource allocation challenges through a variety of approaches. As for convex optimization, Luo~\textit{et~al.}~\cite{luo2020hfel} considered the hierarchical federated edge learning (HFEL) in wireless communication, the minimization of system overhead is addressed separately using the Karush-Kuhn-Tucker (KKT) method in convex optimization for both single and multi-edge server scenarios. Yu~\textit{et~al.}~\cite{yu2020joint} utilized another classical method in optimization, successive convex approximation (SCA) transforming the non-convex problem, that performs resource allocation under constraints on UAV's battery size and service quality, aiming to minimize system latency by transforming the non-convex problem. Additionally, some works also tried to use deep learning strategies for non-convex or complex sequential problems. For instance, He~\textit{et~al.}~\cite{he2019joint} solved the problem of channel assignment under Non-orthogonal multiple access (NOMA) with deep reinforcement learning (DRL), which further utilized attention-based neural network for better performance. Furthermore, Guo~\textit{et~al.}~\cite{guo2020joint} considered a more sophisticated approach to orchestrate multiple agents for jointly solving the handover control and power allocation problem. However, many previous works address the resource allocation problem in various scenarios, but none of them jointly consider such a problem based on LM and related fine-tuning issues. 

\subsection{Large model with mobile edge computing.} \label{rela_2}
Integrating LM with mobile edge computing (MEC) has become a pivotal area of research. Liu~\textit{et~al.}~\cite{shen2023large} explored the use of Generative Pretrained Transformers (GPT) in edge AI systems. Their framework employs a cloud-edge-client hierarchical architecture, where GPT, located in the cloud, coordinates with AI models on edge servers and devices. This system efficiently processes user requests in natural language, demonstrating GPT's potential in autonomously organizing and optimizing edge AI models. Tian~\textit{et~al.}~\cite{tian2022fedbert} proposed a novel learning approach named FedBERT. This method combines federated learning with split learning for pre-training BERT in a federated manner. FedBERT addresses the challenge of pre-training LMs and ensures that sensitive local data of clients are not communicated. Furthermore, Jiang~\textit{et~al.}~\cite{jiang2023high} focused on enhancing the efficiency of large models in edge computing scenarios by balancing computational loads between edge devices and cloud servers and Dong~\textit{et~al.}~\cite{dong2023lambo} highlighted the use of these models in improving the performance and capabilities of edge computing systems. Although many studies consider deploying LM at the edge, they don't address the potential resource allocation issues with optimization accordingly, like the ones specified in our work.


\subsection{Secure wireless communication under physical layer security} \label{rela_3}
Several studies have explored physical layer security with varying numbers of eavesdroppers. Cui~\textit{et~al.}~\cite{cui2019secure} implemented the Alternative Optimization algorithm to jointly optimize the transmit beamforming vector at the base station and the phase elements at the Intelligent Reflecting Surface (IRS), aiming to maximize the secrecy rate. Considering multi-eavesdroppers, Yang~\textit{et~al.}~\cite{yang2020deep} proposed a RL-based secure beamforming approach to jointly optimize users secrecy rate and quality of service (QoS). However, to our knowledge, none specifically address secure communication during the fine-tuning process to safeguard both the model and user security, as articulated in our paper.

\section{System Model and Problem Formulation} \label{sec_sys}
In this section, we formally describe our system model from the holistic to the particular. We will first introduce the offsite-tuning model as the foundation of the entire system, followed by the specific implementation model on the large model owner (LMO) and data owners (DOs) sides, in addition to the transmission model. At the end of this section, we give the comprehensive objective function for the communication model.

\subsection{System Model}
We consider a large model (LM) based secure mobile edge computing system, as shown in Fig.~\ref{fig_system}, where the LMO, which also be referred to as the LM server, positioned at the edge side managing its downlink channel with all $N$ single-antenna mobile users, also designated as DOs, with each DO denoted by $n \in \mathcal{N}=\{1, 2, \dots, N\}$.

\textbf{Offsite-tuning Model.} To preserve both parties’ privacy and enhance efficiency and effectiveness without disrupting the operational system, we utilize offsite-tuning in our system~\cite{xiao2023offsite}. Given a pre-trained Transformer backbones LM $\mathcal{M}$ deployed on the edge server side, encompassing a total of $\mathcal{L}$ layers (i.e., model depth). We divide the $\mathcal{M}$ into two distinct components: a small, trainable adapter $\mathcal{A}_n$, which is used for downstream adaptation for $n$-th DO, and a remaining frozen portion of the model, denoted as emulator $\mathcal{E}_n$. In the selection of adapters, the paramount consideration for the adapter transmitted to the downstream tasks is its suitability across various downstream tasks, which is influenced by the extraction of diverse combinations of shallow and deep layers from the Transformer. We employ the sandwich structure proven superior to single-layer configurations in~\cite{xiao2023offsite}, where $l_{n}^a$ layers are selected simultaneously from both the top and bottom layers of $\mathcal{M}$ (i.e., $\mathcal{M} = \mathcal{A}_{n}^1\circ\mathcal{E}_n\circ\mathcal{A}_{n}^2$ and $\mathcal{A}_{n}^1+\mathcal{A}_{n}^2=l_{n}^a$). A variable called adapter extraction rate denoted as $\boldsymbol{\phi} = [\phi_1, \phi_2, \dots, \phi_N]$ for each DO, $n \in \mathcal{N}$ is introduced to denote the number of layers selects from both top and bottom layers of $\mathcal{M}$. To ensure the feasibility of offsite tuning and the overall system efficiency, $\phi_n \in (0, \frac{2}{L}\times \chi^a_n]$, where $\chi^a_n, n \in \mathcal{N}$ is the confidentiality coefficient for different downstream tasks, $l_{n}^a = 2\phi_n$ and $\frac{2}{L}\times \chi^a_n = \phi_{n, max}$. For convenience, $\mathcal{A}_{n}^1, \mathcal{A}_{n}^2$ are consolidated into $\mathcal{A}_n$, such that $\forall n \in \mathcal{N}, \mathcal{M}=[\mathcal{A}_n,\mathcal{E}_n]$ is the concatenation of these two parts.

For the remaining frozen component emulator $\mathcal{E}_n$, it is imperative that the data contained within $\mathcal{E}_n$ are not overly precise, as such precision could inadvertently reveal information about the original model. To address this, we utilize the layer drop technique~\cite{sajjad2023effect} to uniformly select a subset of layers from $\mathcal{E}_n$, the dropped emulator denoted as $\mathcal{E}^*_n$, for utilization. A discrete variable layer retention rate $\chi^e_n, n \in \mathcal{N}$, is introduced and selected from the interval $[0.1, 0.3, 0.5, 0.7, 0.9]$ to govern the number of layers $l_{n}^e$ present in $\mathcal{E}^*_n$, where $l_{n}^e = \chi^e_n (1-2\phi_n)$. Emulator $\mathcal{E}^*_n$ encompasses the gradient information pertinent to the adapter updates, while these gradient details must not be excessively precise, lest they lead to the inadvertent disclosure of information pertaining to the original model. Furthermore, an emulator characterized by a smaller volume of data is more conducive to system efficiency~\cite{xiao2023offsite}. We aim to find a balance between these requirements. First, we calculate the loss between two emulators using the Mean Square Error (MSE) to ensure $\mathcal{E}^*_n$ encapsulates as much of the gradient information necessary for training as possible. Given that the discrepancy between $\mathcal{E}_n$ and $\mathcal{E}^*_n$ is primarily influenced by $\chi^e_n$, we define the loss as a function solely dependent $\chi^e_n$, which can be expressed as $\mathcal{L}_{MSE}(\chi^e_n) = \frac{1}{N}\sum_{i=1}^N ||\mathcal{E}^*_n - \mathcal{E}_n||^2$. Furthermore, the efficiency can be defined as a function that decreases with an increase in the number of layers within $\mathcal{E}^*_n$, which is correlated with $\phi_n$ and $\chi^e_n$, the efficiency can be represented as $\mathcal{G}(\phi_n, \chi^e_n) = \log(1 + \frac{\phi_n}{\chi^e_n})$, with a designed minimal fine-tuning efficiency $\mathcal{G}(\phi_n, \chi^e_n) \geq E^{\text{eff}}_{n, \text{min}}$.

\textbf{FDMA.} In our paper, we implement Frequency Division Multiple Access (FDMA) for communication between devices and the base station. FDMA is known for its simplicity and suitability for mobile devices with limited computational capabilities, which aids in managing the bandwidth and transmission power allocation without the complications of interference. We consider uplink communication where each DO trains and uploads adapter $\mathcal{A}_n$ back to LMO. We define $\boldsymbol{b} = [b_1, b_2, \dots, b_N]$, $\boldsymbol{p} = [p_1, p_2, \dots, p_N]$ as the bandwidth and transmission power for each DO communication to LMO and $b_{n,\text{max}}$, $p_{n,\text{max}}$ are used for fair resource allocation~\cite{wang2016resource, chen2019energy}.

\textbf{User training adapter $\mathcal{A}_n$.} Based on the above discussion, in our system, we account for the varying GPU frequencies of each user as their computational resources allocated to fine-tuning the adapter $\mathcal{A}_n$, denoted as $\boldsymbol{f^{\text{DO}}} = [f^{\text{DO}}_1, f^{\text{DO}}_2, \dots, f^{\text{DO}}_N]$ and the user's computational capabilities directly impact the training efficiency of the model's adapter. We posit that for any given model considered in our paper, the model's parameter size is invariably positively correlated with the number of layers within the model. In the Transformer models we employ, each layer typically exhibits a similar structure~\cite{vaswani2017attention} (i.e., the model's parameter size can be estimated using the number of layers $\times$ the average parameter count per layer). Consequently, we can directly utilize the number of layers of each model as a representation of the corresponding size of the model. The computational time for each DO can be represented as 
\begin{equation}
    t^{\text{DO:cmp}}_{n}(f^{\text{DO}}_n, \phi_n) = \frac{\mathcal{C}_n(\phi_n)}{f^{\text{DO}}_n}, \label{time_user_proce}
\end{equation}
where $\mathcal{C}_n(\phi_n)$ is the total GPU cycles required for user $n$'s fine-tuning. For computation function $\mathcal{C}_n(\phi_n)$, the number of computation cycles required for user training increases with $\phi_n$ since more trainable parameters in $\mathcal{A}_n$. In our definition, the computation function $\mathcal{C}_n(\phi_n)$ can be formalized as
\begin{equation}
    \mathcal{C}_n(\phi_n)  = C_1 \phi_n ^{C_2}, \forall n \in \mathcal{N},  \label{c_n}
\end{equation}
where $C_1>0$, $C_2>1$. Based on the user's computation discussion~\cite{shen2019computation}, the energy consumed for DO training can be represented as 
\begin{equation}
    E^{\text{DO:cmp}}_{n}(f^{\text{DO}}_n, \phi_n) =  k_n \mathcal{C}_n(\phi_n) (f^{\text{DO}}_n)^2 ,\label{ener_user_process}
\end{equation}
where $k_n$ is the effective switched capacitance for the user. 

\textbf{User sending back adapter to server.} 
After DO's processing, the user transmits the fine-tuned adapter back to the LMO deployed at the edge server side. According to Shannon's formula, the transmission rate for each user can be expressed as 
\begin{equation}
    r_{n}(b_{n}, p_n) = b_{n}\log_2(1+\frac{g_{n}p_n}{\sigma_{n}^2 b_{n}}), \label{rate_user_send}
\end{equation}
where $\sigma_n^2$ is the power spectral density of Gaussian noise, and $g_n$ is the channel attenuation from DO to LMO.

\begin{figure}[h]
    \centering
    \includegraphics[scale=0.35]{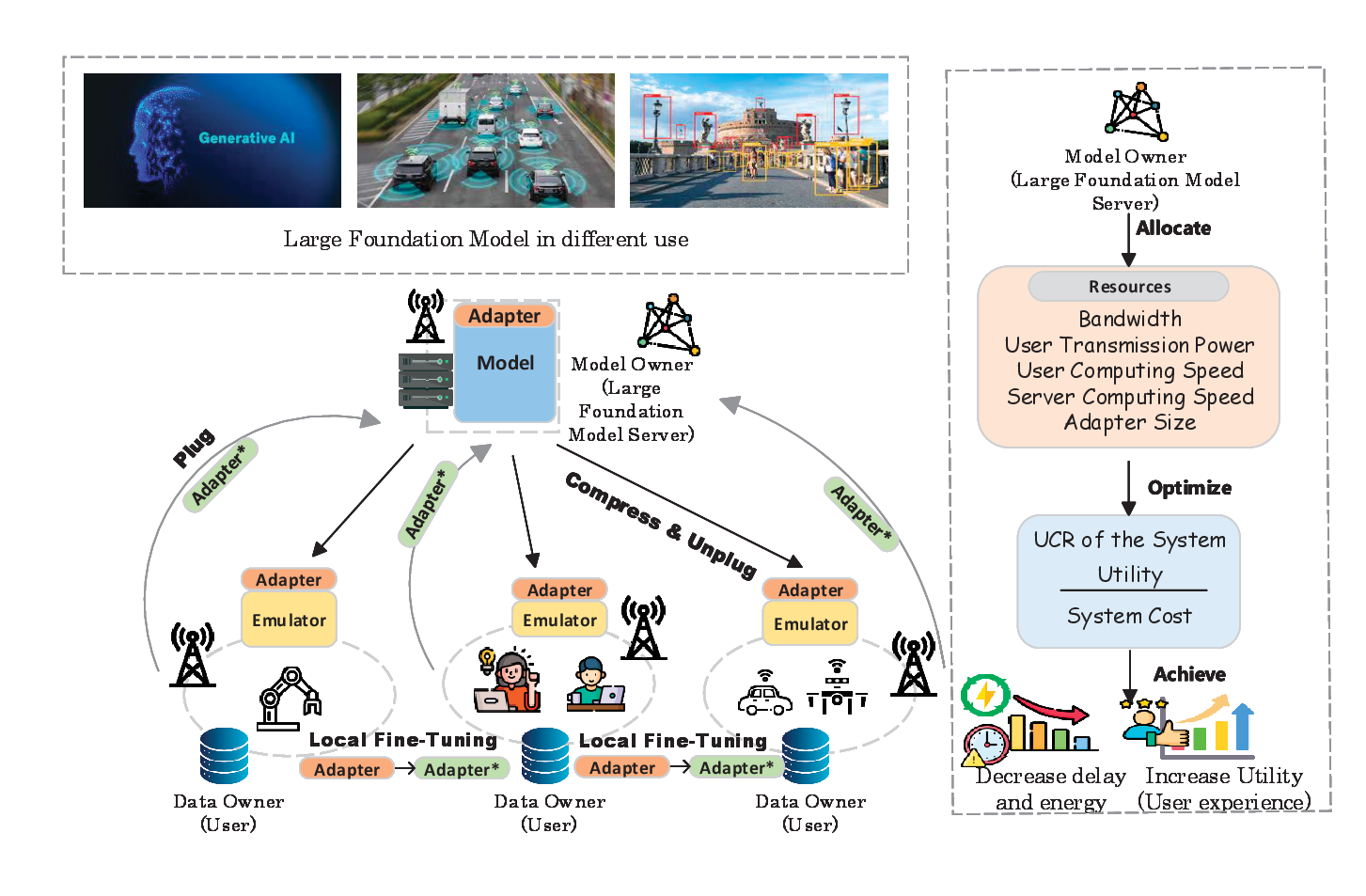} \vspace{-1cm}
    \caption{Optimizing the UCR of a large model system with $N$ data owners and a model owner through joint resource allocation optimization.} \vspace{-0.5cm}
    \label{fig_system}
\end{figure}

We account for a potential unauthorized eavesdropper $E^e$ who aims to eavesdrop on any of the data streams from the DOs and threaten our model. Given the utilization of FDMA in our system, the eavesdropper's ability to intercept the user information is linked to the user's allocated bandwidth and power (i.e., the eavesdropper can disrupt the user's communication by leveraging the bandwidth allocated to the user and the eavesdrop power equals the transmission power of DOs~\cite{yang2020deep}.) The eavesdropping rate for user $n$ is represented as
\begin{equation}
    r_{e}(b_{n}, p_n) = b_n \log_2(1+\frac{g_{n, e}p_n}{\sigma_{e}^2 b_{n}}), \label{rate_user_eav}
\end{equation}
where $\sigma_e^2$ is the power spectral density of Gaussian noise of $E^e$ and $g_{n, e}$ is the channel gain from each DO to $E^e$, respectively.

We use $r_{s,n}$ to represent the rate at which data can be securely transmitted without being intercepted by $E^e$ (i.e., the individual secrecy rate). Meanwhile, we can quantify the impact of eavesdropping on the user's transmission. The secrecy rate of user $n$ is defined as
\begin{equation}
    r_{s,n} (b_{n}, p_n) = \bigg[ r_{n}(b_{n}, p_n) - r_{e}(b_{n}, p_n) \bigg]^+, \label{rate_sec}
\end{equation}
where $[z]^+ = \text{max}(0, z)$. Once $r_{s,n} (b_{n}, p_n) = 0$, the transmission is considered as failure with no subsequent calculations. With the valid secrecy rate, the transmission time for uplink communication is 
\begin{equation}
    t^{\text{com}}_{n}(b_{n}, p_{n}, \phi_n) = \frac{\phi_n \mathcal{L} w}{r_{s,n}(b_n, p_n)}, \label{time_user_trans}
\end{equation}
where $w$ is the number of bits used to represent each parameter. The energy consumed for transmission can be represented as
$E^{\text{com}}_{n} = p_n T^{\text{com}}_{n}$, which is 
\begin{equation}
    E^{\text{com}}_{n}(b_n, p_{n}, \phi_n) = p_n \frac{\phi_n \mathcal{L} w}{r_{s,n}(b_n, p_n)}. \label{ener_user_trans}
\end{equation}

\textbf{Server integrating and processing the received adapter.}
At the server side, let $\boldsymbol{f^{\text{LMO}}} = [f^{\text{MO}}_1, f^{\text{LMO}}_2, \dots, f^{\text{DO}}_N]$ represent the GPU frequency (i.e., the computational resources) processing the transmitted back adapters from user $n$. Considering the parameter size of the server received adapter $\mathcal{A}_n$, the processing time at the server is 
\begin{equation}
    t^{\text{LMO:cmp}}_{}(f^{\text{LMO}}_n, \phi_n)  =\frac{\mathcal{B}_n(\phi_n)}{f^{\text{LMO}}_n} ,\label{time_server_proce}
\end{equation}
where $\mathcal{B}_n(\phi_n)$ is the total GPU cycles required for processing user $n$'s adapter. On the server side, the LMO integrated the adapter into the original model. Given that the parameter volume in the adapter $\mathcal{A}_n$ is much smaller than that in the emulator $\mathcal{E}_n$~\cite{sajjad2023effect}, the computational load at the model owner's end is primarily determined by the parameters in the $\mathcal{E}_n$ (not $\mathcal{E}^*_n$, since $\mathcal{E}^*_n$ is a subset of $\mathcal{E}_n$ and mainly used for DO to fine-tuning). Consequently, the computation function $\mathcal{B}_n(\phi_n)$ represents the number of computation cycles required for the LMO to integrate and process the data, decreases with $\phi_n$. Same as (\ref{c_n}), $\mathcal{B}_n(\phi_n)$ can be formalized as
\begin{equation}
    \mathcal{B}_n(\phi_n) = C_3 {\phi_n} ^{-C_4},
    \label{n_n}
\end{equation}
where $C_3, C_4 >0$. The energy consumed for LMO processing is 
\begin{equation}
    E^{\text{LMO:cmp}}_{n}(f^{\text{LMO}}_n, \phi_n) = k_m \mathcal{B}_n(\phi_n) (f^{\text{LMO}}_n)^2. \label{ener_server_process}
\end{equation}
In conclusion, the overall complete time of LMO completes user $n$'s  task, including both computation and
computation is 
\begin{align}
&t_n (b_n, p_n, \phi_n, f^{\text{DO}}_n, f^{\text{LMO}}_n) \label{time_total_n}\\ &= 
t^{\text{DO:cmp}}_{n}(f^{\text{DO}}_n, \phi_n) +  t^{\text{com}}_{n}(b_{n}, p_{n}, \phi_n)+  t^{\text{LMO:cmp}}_{}(f^{\text{LMO}}_n, \phi_n). \notag 
\end{align}
Then, the system delay is defined as the maximum delay experienced by any user in the network, which can be represented as
\begin{equation}
 \mathcal{T}(\boldsymbol{b},  \boldsymbol{p}, \boldsymbol{\phi}, \boldsymbol{f^{\text{DO}}}, \boldsymbol{f^{\text{MO}}})=  \text{max}_{n \in \mathcal{N}} \ t_n (b_n, p_n, \phi_n, f^{\text{DO}}_n, f^{\text{MO}}_n). \label{time_total}
\end{equation}
The total communication and computation energy consumption of the system is
\begin{align}
\mathcal{E} (\boldsymbol{b},  \boldsymbol{p}, \boldsymbol{\phi}, &\boldsymbol{f^{\text{DO}}}, \boldsymbol{f^{\text{LMO}}}) = \sum_{n \in \mathcal{N}}  E^{\text{DO:cmp}}_{n}(f^{\text{DO}}_n, \phi_n) \label{energy_total}\\& + \sum_{n \in \mathcal{N}} E^{\text{com}}_{n}(b_n, p_{n}, \phi_n) + \sum_{n \in \mathcal{N}} E^{\text{LMO:cmp}}_{n}(f^{\text{LMO}}_n, \phi_n). \notag 
\end{align}
The total cost of the system is determined by a weighted sum, encompassing both the system delay Eq.~(\ref{time_total}), and the energy consumption Eq.~(\ref{energy_total}):
\begin{align}
    \mathcal{S}(\boldsymbol{b},  &\boldsymbol{p}, \boldsymbol{\phi}, \boldsymbol{f^{\text{DO}}}, \boldsymbol{f^{\text{LMO}}}) \label{cost_total} \\&= c_t \mathcal{T}(\boldsymbol{b},  \boldsymbol{p}, \boldsymbol{\phi}, \boldsymbol{f^{\text{DO}}}, \boldsymbol{f^{\text{LMO}}}) + c_e \mathcal{E} (\boldsymbol{b},  \boldsymbol{p}, \boldsymbol{\phi}, \boldsymbol{f^{\text{DO}}}, \boldsymbol{f^{\text{LMO}}}). \notag
\end{align}

\textbf{Utility.} Based on the definition of~\cite{zhao2023human}, for each user $n$, we formulate the service experience score, which represents the `effectiveness' of the users as a function of $b_n, \phi_n$ and $f_n^{\text{DO}}$: $U_n(b_n, \phi_n, f_n^{\text{DO}})$, satisfying Assumption 1 below.
\begin{assumption} \label{assum_utility_def}
    Utility function $U_n(x)$ is non-decreasing in $x$ and concave in $x$, also twice differentiable.
\end{assumption}
\begin{proof}
    See Appendix~\ref{appen_assum_utility_def}.
\end{proof}
Following the above assumption, we can formulate the utility function as the normalization of three parameters, which is relative to their maximum values. The user score (utility) function can be formulated as
\begin{equation}
    U_n(b_n, \phi_n, f_n^{\text{DO}}) = \varpi \ln \big({1+ \phi_n + \frac{f_n^{\text{DO}}}{f_{n, \text{max}}^{\text{DO}}}} +  \frac{b_n}{b_{n,\text{max}}}\big), \label{uti_n}
\end{equation}
where $\varpi$ is used for normalization. This utility function demonstrates efficacy and sensitivity across the entire spectrum of optimal value ranges. The system utility, defined as the sum of all $N$ users' utilities, is given by
\begin{equation}
   \mathcal{U}(\boldsymbol{b}, \boldsymbol{\phi}, \boldsymbol{f^{\text{DO}}}) = \sum_{n \in \mathcal{N}} U_n (b_n, \phi_n, f_n^{\text{DO}}). \label{total_utility}
\end{equation}
To accomplish maximum economic efficiency of the system,  we aim to maximize the system's utility while minimizing the cost, the utility-consumption ratio (UCR) problem $\mathbb{P}_1$ can be formulated as
\begin{subequations} \label{prob_1}
\begin{align}
\mathbb{P}_1: \ &\max_{\boldsymbol{b},  \boldsymbol{p}, \boldsymbol{\phi}, \boldsymbol{f^{\text{DO}}}, \boldsymbol{f^{\text{MO}}}}  \frac{ \mathcal{U}(\boldsymbol{b}, \boldsymbol{\phi}, \boldsymbol{f^{\text{DO}}})}{\mathcal{S}(\boldsymbol{b},  \boldsymbol{p}, \boldsymbol{\phi}, \boldsymbol{f^{\text{DO}}}, \boldsymbol{f^{\text{MO}}})} 
\tag{\ref{prob_1}} \\
\text{subject to:} \notag \\
& b_{n} \leq b_{n,\text{max}}, \forall n \in \mathcal{N}, \label{cons_band}\\
& p_n \leq p_{n, \text{max}}, \forall n \in \mathcal{N}, \label{cons_pow_user} \\
& \phi_n \in (0,\phi_{n, \text{max}}], \forall n \in \mathcal{N}, \label{cons_data_sele} \\
& f^{\text{DO}}_n \leq f^{\text{DO}}_{n, \text{max}}, \forall n \in \mathcal{N}, \label{cons_freq_user} \\
& \sum_{n \in \mathcal{N}} f^{\text{MO}}_n \leq f^{\text{MO}}_{\text{max}}, \label{cons_freq_server} \\
& \mathcal{G}(\phi_n, \chi^e_n) \geq E^{\text{eff}}_{n, \text{min}}, \forall n \in \mathcal{N}, \label{cons_rate_effi} \\
& \mathcal{L}_{MSE}(\chi^e_n) \leq \text{Loss}_{n, \text{max}}, \forall n \in \mathcal{N}, \label{cons_loss} \\
& b_n, p_n \geq 0, \chi^e_n \in [0.1, 0.3, 0.5, 0.7, 0.9], \forall n \in \mathcal{N}. \label{cons_geq}
\end{align}
\end{subequations}
In Problem $\mathbb{P}_1$,~(\ref{cons_band}),~(\ref{cons_pow_user}), and (\ref{cons_freq_user}) ensures the bandwidth, power, and GPU frequencies for each transmit user must not exceed their respective maximum limits, which means the users operate within their resource capacities.~(\ref{cons_data_sele}) gives the limitation of the extraction rate.~(\ref{cons_freq_server}) limits the total resources LMO used for computation, and~(\ref{cons_rate_effi})~(\ref{cons_loss}) achieved a balance between the effectiveness and the computational efficiency of $\mathcal{E}^*_n$.

\section{Solution of the joint optimization problem} \label{sec_solu}
The problem defined in~(\ref{prob_1}) is a non-convex fractional programming problem classified as NP-hard~\cite{freund2001solving}. Since in~(\ref{prob_1}), the numerator maintains non-negativity, the denominator is strictly positive, and both are continuous functions, we use the Dinkelbach transform~\cite{shen2018fractional} to convert problem $\mathbb{P}_1$ to $\mathbb{P}_2(y)$ with introduce an additional variable $y$. Let $H_{\mathbb{P}_2}$ denote the objective function, then $H_{\mathbb{P}_2}$ and $\mathbb{P}_2$ are
\begin{align}
H_{\mathbb{P}_2}(\boldsymbol{b},  &\boldsymbol{p}, \boldsymbol{\phi}, \boldsymbol{f^{\text{DO}}}, \boldsymbol{f^{\text{MO}}}\ |\ y):\\ & = \text{numerator of (\ref{prob_1})}  - y \cdot \text{denominator of (\ref{prob_1})} \notag \\
\mathbb{P}_2(y): & \max_{\boldsymbol{b},  \boldsymbol{p}, \boldsymbol{\phi}, \boldsymbol{f^{\text{DO}}}, \boldsymbol{f^{\text{MO}}}} H_{\mathbb{P}_2}(\boldsymbol{b},  \boldsymbol{p}, \boldsymbol{\phi}, \boldsymbol{f^{\text{DO}}}, \boldsymbol{f^{\text{MO}}}, T \ |\ y) \label{prob_sub_1_3} \\
\text{subject to:} &~\text{(\ref{cons_band})} - \text{(\ref{cons_geq}).}\notag
\end{align}

{\color{red}



}

To address problem $\mathbb{P}_2(y)$, we define the variables [$\boldsymbol{b},  \boldsymbol{p}, \boldsymbol{\phi}, \boldsymbol{f^{\text{DO}}}, \\ \boldsymbol{f^{\text{MO}}}, T$] as $\boldsymbol{z}$ and express the objective function of $\mathbb{P}_2(y)$ as $ U(\boldsymbol{z}) - y \cdot S(\boldsymbol{z}) $. At the beginning of the iteration, we initiate with a feasible $\boldsymbol{z}^{(0)} $ and set $y^{(0)} = \frac{U(\boldsymbol{z}^{(0)})}{S(\boldsymbol{z}^{(0)})}$. Subsequently, we solve $ \mathbb{P}_2(y^{(0)})$, obtaining the solution $\boldsymbol{z}^{(1)}$, and update $y^{(1)} = \frac{U(\boldsymbol{z}^{(1)})}{S(\boldsymbol{z}^{(1)})}$. This iterative method proceeds as follows: for the $(k+1)$-th iteration, $y^{(k)}$ is set as $\frac{U(\boldsymbol{z}^{(k)})}{S(\boldsymbol{z}^{(k)})}$, and $\boldsymbol{z}^{(k+1)} $ is derived by solving $ \mathbb{P}_2(y^{(k)}) $. This procedure ensures convergence and maintains optimality. 

After addressing the Dinkelbach transformation, $H_{\mathbb{P}_2}$ still contains the `max' function of delay, we introduce an auxiliary variable $T$ to circumvent this difficulty. Furthermore, with known $y$, Problem $\mathbb{P}_2(y)$ still contains six variables, we decompose the problem into two subproblems, one has $\boldsymbol{b},  \boldsymbol{p}, \boldsymbol{f^{\text{DO}}}, \boldsymbol{f^{\text{MO}}}, T$ as optimization variables and the other has $\boldsymbol{\phi}$ as optimization variables.
\paragraph{Subproblem 1}
\begin{subequations} \label{prob_sub_1_1} 
\begin{align}
&\max_{\boldsymbol{\phi}} H_{\mathbb{P}_2}(\boldsymbol{b},  \boldsymbol{p}, \boldsymbol{\phi}, \boldsymbol{f^{\text{DO}}}, \boldsymbol{f^{\text{MO}}}, T \ |\ y) \tag{\ref{prob_sub_1_1}} \\
    &  ~\text{(\ref{cons_data_sele}),}~\text{(\ref{cons_rate_effi}).}  \notag
\end{align}
\end{subequations}
\paragraph{Subproblem 2}
\begin{subequations} \label{prob_sub_2_1}
\begin{align}
&\max_{\boldsymbol{b},  \boldsymbol{p}, \boldsymbol{f^{\text{DO}}}, \boldsymbol{f^{\text{MO}}}, T} H_{\mathbb{P}_2}(\boldsymbol{b},  \boldsymbol{p}, \boldsymbol{\phi}, \boldsymbol{f^{\text{DO}}}, \boldsymbol{f^{\text{MO}}}, T \ |\ y) \tag{\ref{prob_sub_2_1}} \\
&~\text{(\ref{cons_band}),}~\text{(\ref{cons_pow_user}),}~\text{(\ref{cons_freq_user}),}~\text{(\ref{cons_freq_server}),}~\text{(\ref{cons_geq}).} \notag\\
& T \geq t_n (b_n, p_n, \phi_n, f^{\text{DO}}_n, f^{\text{MO}}_n), \forall n \in \mathcal{N}\label{cons_new_time}.
\end{align}
\end{subequations}
The methodology for solving the problem is outlined as follows. At the beginning of each iteration, first calculate the minimal $\chi^{e^{-}}_{n}$ following constraint (\ref{cons_loss}), then we utilize Alternating Optimization (AO) with an initial value of $\boldsymbol{\phi}$ and alternately optimize (\ref{prob_sub_1_1}) and (\ref{prob_sub_2_1}). Additionally, the maximal value of $\chi^e_n$ can be calculated using $\chi^{e^{+}}_{n} = \big(\mathcal{G}(\phi_n, \chi^{e^{+}}_n) = E_{n, \text{min}}^{\text{eff}}\big)^{-1}$. Ultimately, the optimal $\chi^{e}_{n} = \text{ROUND}\big(\text{min}\big[\chi^{e^{-}}_{n}, \chi^{e^{+}}_{n} \big]$\big) to the nearest discrete value, as under identical transmission conditions, our primary concern lies with data security rather than exhaustive fine-tuning, which can be compensated for in subsequent rounds of transmission and training.

\subsection{Solution to Subproblem 1}
With fixed $\boldsymbol{b},  \boldsymbol{p}, \boldsymbol{f^{\text{DO}}}, \boldsymbol{f^{\text{MO}}}, T $ , we first optimize $\phi_n$ and $y$. With fixed initial $\chi^{e^{-}}_{n}$ and the monotonic increasing property of $\mathcal{G}(\phi_n, \chi^{e}_n)$ with respect to $\phi_n$, a minimal extraction rate of each DO can be derived from (\ref{cons_rate_effi}), which is
\begin{equation}
    \phi_n \geq \phi_{n,\text{min}}, \ \mathcal{G}(\phi_{n, \text{min}}, \chi^{e^{-}}_n) =  E^{\text{eff}}_{n, \text{min}}\label{cons_1_4_cha_effi}
\end{equation}
Subproblem (\ref{prob_sub_1_1}) can be reformulated as
\begin{subequations} \label{prob_4}
\begin{align}
\mathbb{P}_3: \min_{\boldsymbol{\phi}} & - \mathcal{U}(\boldsymbol{b}, \boldsymbol{\phi}, \boldsymbol{f^{\text{DO}}}) + y \cdot (c_e \mathcal{E} (\boldsymbol{b},  \boldsymbol{p}, \boldsymbol{\phi}, \boldsymbol{f^{\text{DO}}}, \boldsymbol{f^{\text{MO}}}) + c_t T)
\tag{\ref{prob_4}} \\
\text{subject to:} \notag \\
& \phi_{n,\text{min}} \leq \phi_n \leq \phi_{n,\text{max}}, \forall n \in \mathcal{N}, \label{cons_1_5_sele}\\
&~\text{(\ref{cons_new_time})} \notag
\end{align}
\end{subequations}
Since $\mathbb{P}_3$ is convex, the Karush-Kuhn-Tucker (KKT) conditions are the sufficient and necessary optimality
conditions for finding the optimal solution. We can write down the partial
Lagrangian function of problem  $\mathcal{L}_1(\boldsymbol{\phi}, \boldsymbol{\eta} \ | \ y, \boldsymbol{b},  \boldsymbol{p}, \boldsymbol{f^{\text{DO}}}, \boldsymbol{f^{\text{MO}}}, T)=$
\begin{align}
   &-\sum_{n \in \mathcal{N}}\ln \big({\theta_n + \frac{f_n^{\text{DO}}}{f_{n, \text{max}}^{\text{DO}}}} +  \frac{b_n}{b_{n,\text{max}}}\big) + yc_e\sum_{n \in \mathcal{N}}[k_n C_1 \phi_n ^{C_2} (f^{\text{DO}}_n)^2 \label{lag_1}\\
   &+ p_n \frac{\phi_n \mathcal{L} w}{r_{s,n}} + k_m C_3 \phi_n ^{-C_4} (f^{\text{MO}}_n)^2] \notag\\
   &+ \sum_{n \in \mathcal{N}} \eta_n \ [\text{max} \{\frac{C_1 \phi_n ^{C_2}}{f^{\text{DO}}_n} + \frac{\phi_n \mathcal{L} w}{r_{s,n}} + \frac{C_3 \phi_n ^{-C_4}}{f^{\text{MO}}_n} \} - T], \notag 
\end{align}
where $\boldsymbol{\eta} = [\eta_1, \eta_2, \dots, \eta_N]$ are Lagrange multipliers associate with constraint. \textbf{Theorem 1} is given to find the optimal $\boldsymbol{\phi}$. 
\begin{theorem} \label{theo_sub_1_solu}
The optimal solution of problem $\mathbb{P}_3$ is 
\begin{equation}
    \phi^*_n = \text{min}\{\text{max}\{\phi_{n, \text{min}}, \widetilde{\phi}_n(\eta_n \ | \ y, \boldsymbol{b},  \boldsymbol{p}, \boldsymbol{f^{\text{DO}}}, \boldsymbol{f^{\text{MO}}}, T)\}, \phi_{n, \text{max}}\}
\end{equation}
where $\eta_n$ needs to satisfy
\begin{align}
    \text{max} &\{\frac{C_1 (\widetilde{\phi}_n(\eta_n) |_{\phi_{n,\text{min}}}^{\phi_{n, \text{max}}} ) ^{C_2}}{f^{\text{DO}}_n} + \frac{(\widetilde{\phi}_n(\eta_n) |_{\phi_{n,\text{min}}}^{\phi_{n, \text{max}}} ) \mathcal{L} w}{r_{s,n}} \\ &+ \frac{C_3 (\widetilde{\phi}_n(\eta_n) |_{\phi_{n,\text{min}}}^{\phi_{n, \text{max}}} ) ^{-C_4}}{f^{\text{MO}}_n} \} = T \notag
\end{align}
with $\widetilde{\phi}_n(\eta_n) |_{a}^{b} := \text{max}(a, \text{min}( \widetilde{\phi}_n(\eta_n) ,b)) $.
\end{theorem}
\begin{proof}
    See Appendix~\ref{appen_sub_1_solu}.
\end{proof}
\subsection{Solution to Subproblem 2}
With fixed $\boldsymbol{\phi}$, we need to handle the sum-of-ratio function $\sum_{n \in \mathcal{N}}\frac{p_n \phi_n \mathcal{L}w}{r_{s,n}}$. We utilize the fractional programming (FP) technique proposed in~\cite{zhao2023human}, Problem (\ref{prob_sub_2_1}) can be converted into
\begin{subequations} \label{prob_5}
\begin{align}
    \mathbb{P}_4(\boldsymbol{x}, \boldsymbol{\phi}, y): &\max_{\boldsymbol{b},  \boldsymbol{p}, \boldsymbol{f^{\text{DO}}}, \boldsymbol{f^{\text{MO}}}, T} F(\boldsymbol{b},  \boldsymbol{p}, \boldsymbol{f^{\text{DO}}}, \boldsymbol{f^{\text{MO}}}, T) \tag{\ref{prob_5}}\\&- yc_e \cdot \sum \bigg\{ [p_n \phi_n \mathcal{L} w]^2x_n +  \frac{1}{4[r_{s,n}]^2 x_n}\bigg\} \notag\\
    \text{subject to:} &~\text{(\ref{cons_band}),}~\text{(\ref{cons_pow_user}),}~\text{(\ref{cons_freq_user}),}~\text{(\ref{cons_freq_server}),}~\text{(\ref{cons_geq}),}~\text{(\ref{cons_new_time}),} \notag
\end{align}
\end{subequations}
where $F(\boldsymbol{b},  \boldsymbol{p}, \boldsymbol{f^{\text{DO}}}, \boldsymbol{f^{\text{MO}}}, T) = \mathcal{U}(\boldsymbol{b}, \boldsymbol{\phi}, \boldsymbol{f^{\text{DO}}}) - y \bigg[c_e \sum_{n \in \mathcal{N}}\bigg(k_n \mathcal{C}_n(\phi_n) \\ (f^{\text{DO}}_n)^2 + k_m \mathcal{B}_n(\phi_n) (f^{\text{MO}}_n)^2 + c_t T  \bigg]$. The introduced auxiliary variables $\boldsymbol{x}=[x_1, x_2, \dots, x_N] >0 $, and satisfy $[p_n \phi_n \mathcal{L} w]^2x_n +  \frac{1}{4[r_{s,n}]^2 x_n} = [\mathcal{M}_n(\boldsymbol{z})x_n+\frac{\mathcal{P}_n(\boldsymbol{z})}{x_n})]$, where $x_n = \sqrt{\frac{\mathcal{P}_n(\boldsymbol{z})}{\mathcal{M}_n(\boldsymbol{z})}}$. The iteration algorithm of finding $\boldsymbol{x}$ is the same as finding $y$ in (\ref{prob_sub_1_3}).

Since the secrecy rate (\ref{rate_sec}) is neither jointly convex nor concave with respect to $b_n$ and $p_n$ (i.e., the expression of the Hessian matrix of (\ref{rate_sec}) involves a mixture of multi-squared terms and differences of squares, indicating that the expression's sign can vary based on the relative sizes of parameters in $r_{s,n}$, and the presence of both positive and negative components in the numerator suggests that the function may not be uniformly convex or concave across its entire domain). We utilize the Successive Convex Approximation (SCA) method to approximate (\ref{rate_sec}) at ($b_n^{(k+1)}$, $p_n^{(k+1)}$). Assume $b_n^{(k)}$, $p_n^{(k)}$ are the current estimate of $b_n$ and $p_n$ at the $k$-th iteration, we can linearize (\ref{rate_sec}) using the first-order Taylor expansion, the linear approximation can be expressed as 
\begin{align}
    \widetilde{r}_{s,n} (b_{n}, p_n) = & r_{s,n} (b_{n}^{k}, p_n^{k}) + \frac{\partial}{\partial b_n}\big[r_{s,n} (b_{n}^{k}, p_n^{k}) \big](b_n - b_n^{(k)}) \\ & + \frac{\partial}{\partial p_n}\big[r_{s,n} (b_{n}^{k}, p_n^{k}) \big](p_n - p_n^{(k)}), \notag
\end{align}
where $b_n^{(k)}$, $p_n^{(k)}$ are intermediate values and regraded as constant. The SCA process involves iteratively solving this optimization problem, updating the estimates ($b_n^{(k+1)}$, $p_n^{(k+1)}$), and repeating until convergence. Then we can say $r_{n,s}(b_n, p_n)$ is convex. 

After the FP and SCA, $\mathbb{P}_4$ is a convex problem. The Lagrange function $\mathcal{L}_2(\boldsymbol{b_e}, \boldsymbol{p}, \boldsymbol{f^{\text{DO}}}, \boldsymbol{f^{\text{MO}}}, T, \boldsymbol{\alpha}, \boldsymbol{\beta}, \boldsymbol{\delta}, \zeta, \boldsymbol{\rho} \ | \ \boldsymbol{x}, \boldsymbol{\phi}, y ) = $
\begin{align}
    & - F(\boldsymbol{b},  \boldsymbol{p}, \boldsymbol{f^{\text{DO}}}, \boldsymbol{f^{\text{MO}}}, T) + yc_e \cdot \sum \bigg\{ [p_n \phi_n \mathcal{L} w]^2x_n +  \frac{1}{4[\widetilde{r}_{s,n}]^2 x_n}\bigg\} \label{l_2} \\ & + \sum [\alpha_n \cdot (b_n - b_{n, \text{max}})] +  \sum [\beta_n \cdot (p_n - p_{n, \text{max}})] \notag \\ & +\sum [\delta_n \cdot (f_n^{\text{DO}}-f_{n, \text{max}}^{\text{DO}})] + \zeta \cdot (\sum_{n \in \mathcal{N}} f^{\text{MO}}_n - f^{\text{MO}}_{\text{max}}) \notag \\ 
    &+\sum [\rho_n \cdot (t_n (b_n, p_n, f^{\text{DO}}_n, f^{\text{MO}}_n) - T)]. \notag 
\end{align}
We present the KKT conditions as
\begin{itemize}
\item \textbf{Stationary:} (The value of each equation = 0)
    \begin{equation}
        \frac{\partial \mathcal{L}_2 }{\partial b_n} = -\frac{b_{n,\text{max}}}{b_n} - \frac{yc_e}{2x_n(\widetilde{r}_{s,n})^3}  \frac{\partial \widetilde{r}_{s,n}}{\partial b_n} + \alpha_n - \rho_n \frac{\phi_n \mathcal{L} w}{(\widetilde{r}_{s,n})^2}  \frac{\partial \widetilde{r}_{s,n}}{\partial b_n} \label{kkt_1}
    \end{equation}
    \begin{equation}
          2yc_ex_n(\phi_n \mathcal{L} w)^2 p_n - \bigg(\frac{\partial \widetilde{r}_{s,n}}{\partial p_n} + \rho_n \frac{\phi_n \mathcal{L} w}{(\widetilde{r}_{s,n})^2} \bigg)\frac{yc_e}{2x_n(\widetilde{r}_{s,n})^3}  + \beta_n  \label{kkt_2}
    \end{equation}
    \begin{equation}
         \frac{\partial \mathcal{L}_2 }{\partial f_n^{\text{DO}}} =  -\frac{f_{n,\text{max}}^{\text{DO}}}{f_n^{\text{DO}}} + 2yc_ek_n\mathcal{C}_n(\phi_n)f_n^{\text{DO}} + \delta_n - \rho_n\frac{\mathcal{C}_n( \phi_n)}{(f_n^{\text{DO}})^2} \label{kkt_3}
    \end{equation}
    \begin{equation}
        \frac{\partial \mathcal{L}_2 }{\partial f_n^{\text{MO}}} =  2yc_ek_m\mathcal{B}_n(\phi_n)f_n^{\text{MO}} + \zeta - \rho_n\frac{\mathcal{B}_n(\phi_n)}{(f_n^{\text{MO}})^2} \label{kkt_4}
    \end{equation}
    \begin{equation}
        \frac{\partial \mathcal{L}_2 }{\partial T} = \sum_{n \in \mathcal{N}} \rho_n - yc_t \label{kkt_5}
    \end{equation}
\item \textbf{Complementary Slackness:} $\forall \in \mathcal{N}$
\begin{align}
    &\alpha_n \cdot (b_{n} - b_{n, \text{max}}) = 0 \label{kkt_6}\\ 
    & \beta_n \cdot (p_n - p_{n, \text{max}}) = 0,   \label{kkt_7}\\
    & \delta_n \cdot (f_n^{\text{DO}}-f_{n, \text{max}}^{\text{DO}}) = 0, \label{kkt_8} \\
    & \zeta \cdot (\sum_{n \in \mathcal{N}} f^{\text{MO}}_n - f^{\text{MO}}_{\text{max}}) = 0 \label{kkt_9}\\
    & \rho_n \cdot (t_n (b_n, p_n, f^{\text{DO}}_n, f^{\text{MO}}_n) - T) = 0, \label{kkt_10}
\end{align}
\item \textbf{Dual Feasibility:}
$~\text{(\ref{cons_band}),}~\text{(\ref{cons_pow_user}),}~\text{(\ref{cons_freq_user}),}~\text{(\ref{cons_freq_server}),}~\text{(\ref{cons_geq}),}~\text{(\ref{cons_new_time}).}$
\item \textbf{Primal Feasibility:} $\forall n \in \mathcal{N}$
\begin{align}
\begin{array}{ll}
     &\text{(\ref{kkt_11}a): } \alpha_n \geq 0;\text{(\ref{kkt_11}b): } \beta_n \geq 0; 
     \text{(\ref{kkt_11}c): } \delta_n \geq 0; \\ &\text{(\ref{kkt_11}d): } \zeta \geq 0;
     \text{(\ref{kkt_11}e): } \rho_n \geq 0,
\end{array} \label{kkt_11}
\end{align}
\end{itemize}
where (\ref{kkt_2}) is $\frac{\partial \mathcal{L}_2 }{\partial p_n}$. Using (\ref{kkt_1})-(\ref{kkt_11}), we develop a three-step algorithm with known ``$y, x, \boldsymbol{\phi}$" (we denote as $\Delta$ below). \textbf{Step 1:} Find $(\boldsymbol{f^{\text{DO}}}, \boldsymbol{f^{\text{DO}}}, \boldsymbol{\delta}, \zeta \ |\ \Delta)$ with respect to $\boldsymbol{\rho}$. \textbf{Step 2:} Find $(\boldsymbol{b}, \boldsymbol{p}, \boldsymbol{\alpha}, \boldsymbol{\beta} \ |\ \Delta)$ with respect to $\boldsymbol{\rho}$. \textbf{Step 3:} Find$(T, \boldsymbol{\rho}\ |\ \Delta)$.

\textbf{Step 1.} From (\ref{kkt_3}), we could derive the solution of $f_n^{\text{DO}}$ with represented by $(\delta_n, \rho_n)$ defined as $\widehat{f}_n^{\text{DO}}(\delta_n, \rho_n \ |\ \Delta)$, based on the value of $\widehat{f}_n^{\text{DO}}(\delta_n, \rho_n \ |\ \Delta)$, we have the following discussion:
\begin{itemize}
    \item \textbf{Case 1:} If $\widehat{f}_n^{\text{DO}}(\delta_n, \rho_n \ |\ \Delta) < f_{n,\text{max}}^{\text{DO}}$ 

    The condition of this case can also be represented as  \\$\widehat{f}_n^{\text{DO}}(\delta_n, \rho_n \ |\ \Delta) \geq \sqrt[3]{\frac{\rho_n}{2yc_ek_n}}$. In this case, based on (\ref{kkt_8}), we can simply set $\delta_n = 0$. We have $f_n^{\text{DO}} = \widehat{f}_n^{\text{DO}}(0, \rho_n \ |\ \Delta)$
    \item \textbf{Case 2:} If $\widehat{f}_n^{\text{DO}}(\delta_n, \rho_n \ |\ \Delta) \geq f_{n,\text{max}}^{\text{DO}}$ 

    In this case, the condition can be written as $\widehat{f}_n^{\text{DO}}(\delta_n, \rho_n \ |\ \Delta) < \sqrt[3]{\frac{\rho_n}{2yc_ek_n}}$. Since $\delta >0$, according to (\ref{kkt_8}), $f_n^{\text{DO}} = f_{n, \text{max}}^{\text{DO}}$, substituting it into (\ref{kkt_3}), we can obtain the equation of $\delta_n$ represented by $\rho_n$, specifically $\widehat{\delta}_n (\rho_n \ |\ \Delta)|_{{f_n^{\text{DO}} = f_{n, \text{max}}^{\text{DO}}}} = 1+\rho_n\frac{\mathcal{C}_n(\phi_n)}{(f_{n, \text{max}}^{\text{DO}})^2}  -2yc_ek_n\mathcal{C}_n(\phi_n) f_{n, \text{max}}^{\text{DO}}$.
\end{itemize}
Summarize both cases, and the conclusion could be derived:
\begin{equation}
    \widetilde{f}_n^{\text{DO}}(\rho_n \ |\ \Delta)= \text{min}\{f_{n,\text{max}}^{\text{DO}}, \widehat{f}_n^{\text{DO}}(0, \rho_n \ |\ \Delta)\} ,\label{opt_f_do}
\end{equation}
\begin{align} \label{opt_del}
    \widetilde{\delta}_n(\rho_n \ |\ \Delta)=\left\{
        \begin{array}{ll}
             0, &   \widetilde{f}_n^{\text{DO}}(0, \rho_n \ |\ \Delta) < f_{n,\text{max}}^{\text{DO}}\\
             \widehat{\delta}_n (\rho_n \ |\ \Delta), &  \text{others} 
        \end{array}
    \right. \nonumber
\end{align}
Similarly, The solution of $f_n^{\text{MO}} = \widehat{f}_n^{\text{MO}}(\zeta, \rho_n \ |\ \Delta)$ can be obtained from (\ref{kkt_4}) as a function of $(\zeta, \rho_n)$, the discussion is
\begin{itemize}
    \item \textbf{Case 1.} If $\sum_{n \in \mathcal{N}}\widehat{f}_n^{\text{MO}}(\zeta, \rho_n \ |\ \Delta) < f_{\text{max}}^{\text{MO}}$

    The condition can also be written as $\sum_{n \in \mathcal{N}}\sqrt[3]{\frac{\rho_n}{2yc_ek_m}} \leq f_{\text{max}}^{\text{MO}}$. In this case, according to~(\ref{kkt_4})~(\ref{kkt_9}), $\zeta=0$. We have $f_n^{\text{MO}} = \widehat{f}_n^{\text{MO}}(0, \rho_n \ |\ \Delta)$. 
    \item \textbf{Case 2.} If $\sum_{n \in \mathcal{N}}\widehat{f}_n^{\text{MO}}(\zeta, \rho_n \ |\ \Delta) \geq f_{\text{max}}^{\text{MO}}$

    The condition means $\sum_{n \in \mathcal{N}}\sqrt[3]{\frac{\rho_n}{2yc_ek_m}} > f_{\text{max}}^{\text{MO}}$, since $\zeta >0$, according to (\ref{kkt_9}), $\sum_{n \in \mathcal{N}} f^{\text{MO}}_n = f^{\text{MO}}_{\text{max}}$, substituting it into (\ref{kkt_4}), we can obtain the equation of $\zeta$ represented by $(\rho_n)$, $\widehat{\zeta} (\rho_n \ |\ \Delta)|_{\sum_{n \in \mathcal{N}}f_n^{\text{MO}} = f_{\text{max}}^{\text{MO}}}$, then use (\ref{kkt_4}) get the value of $f_n^{\text{MO}}$, $f_n^{\text{MO}} = \widehat{f}_n^{\text{MO}}(\widehat{\zeta} (\rho_n \ |\ \Delta), \rho_n \ |\ \Delta)$   
\end{itemize}
Summarize both cases, and the conclusion could be derived:
\begin{align}
   & \widetilde{f}_n^{\text{MO}}(\rho_n \ |\ \Delta)  \\& = \text{min}\{\widehat{f}_n^{\text{MO}}(\widehat{\zeta} (\rho_n \ |\ \Delta), \rho_n \ |\ \Delta), \widehat{f}_n^{\text{MO}}(0, \rho_n \ |\ \Delta)\} \notag
\end{align}
\begin{align}
    \widetilde{\zeta}(\rho_n \ |\ \Delta)=\left\{
        \begin{array}{ll}
             0, &   \sum_{n \in \mathcal{N}}\widetilde{f}_n^{\text{MO}}(0, \rho_n \ |\ \Delta) < f_{\text{max}}^{\text{MO}}\\
             \widehat{\zeta}_n (\rho_n \ |\ \Delta), &  \text{others}
        \end{array}
    \right. \nonumber
\end{align}
\textbf{Step 2.} First, from~(\ref{kkt_1}), since all the functions except $\alpha_n$ are negative and the result is $0$, $\alpha_n >0$ and which result in $b_n = b_{n, \text{max}}$ from~(\ref{kkt_6}). Substituting $b_n = b_{n, \text{max}}$ into~(\ref{kkt_1}), we have
\begin{equation}
    \bigg[- \frac{yc_e}{2x_n(\widetilde{r}_{s,n})^3}  \frac{\partial \widetilde{r}_{s,n}}{\partial b_n} - \rho_n \frac{\phi_n \mathcal{L} w}{(\widetilde{r}_{s,n})^2}  \frac{\partial \widetilde{r}_{s,n}}{\partial b_n} \bigg]_{b_n = b_{n, \text{max}}}  + \alpha_n = 1 \label{change_kkt_1}
\end{equation}
The solution of $\alpha_n$ in (\ref{change_kkt_1}) can be denoted as a function of $p_n$ and Lagrange multipliers $\rho_n$ which is $\widehat{\alpha}_n(p_n, \rho_n \ |\ \Delta)$. Substituting $b_n = b_{n, \text{max}}$ into (\ref{kkt_2}), we can obtain the following equation:
\begin{align}
    &\bigg[  -\frac{yc_e}{2x_n(\widetilde{r}_{s,n})^3} \frac{\partial \widetilde{r}_{s,n}}{\partial p_n} - \rho_n \frac{\phi_n \mathcal{L} w}{(\widetilde{r}_{s,n})^2}  \frac{\partial \widetilde{r}_{s,n}}{\partial p_n} \bigg]_{b_n = b_{n, \text{max}}} \label{change_kkt_2}\\
    &+ 2yc_ex_n(\phi_n \mathcal{L} w)^2 p_n + \beta_n = 0. \notag
\end{align}
From (\ref{change_kkt_2}), we could derive the solution of $p_n$  represented by $(\beta_n, \rho_n)$ defined as $ \widehat{p}_n(\beta_n, \rho_n \ |\ \Delta)$. Summarize (\ref{change_kkt_1})(\ref{change_kkt_2}), we can obtain the optimal value of $\boldsymbol{b}^*$, and the expression of $\boldsymbol{p}, \boldsymbol{\alpha}$ are as follows
\begin{align}
    &\left\{\hspace{-2pt}
    \begin{array}{l}
         b_n^* = b_{n, max} \\
         \widehat{p}_n (\beta_n,  \rho_n \ |\ \Delta)= max \{\check{p}_n(\beta,  \rho_n \ |\ \Delta), 0  \}, \\
         \widehat{\alpha}_n (p_n, \rho_n \ |\ \Delta)= \check{\alpha}_n(\widetilde{p}_n (\beta_n,  \rho_n \ |\ \Delta), \rho_n \ |\ \Delta),  
    \end{array}
    \right. \nonumber \\
\end{align}
Since the expression of $\widehat{\alpha}_n (p_n, \rho_n \ |\ \Delta)$ contains $\boldsymbol{p}$, which means we can only analyze and derive the value of $\boldsymbol{\alpha}$ after we get $\boldsymbol{p}$, for $\widehat{p}_n (\beta_n, \rho_n \ |\ \Delta)$, we can have the following discussion
\begin{itemize}
    \item \textbf{Case 1.} If $\widehat{p}_n (\beta_n,  \rho_n \ |\ \Delta) < p_{n, \text{max}}$

    In this case, according to (\ref{kkt_7}), we can simply set $\beta_n = 0$. We have $p_n = \widehat{p}_n(0, \rho_n \ |\ \Delta)$ and constraints (\ref{kkt_2}), (\ref{cons_pow_user}), (\ref{kkt_11}b) are satisfied.
    \item \textbf{Case 2.} If $\widehat{p}_n (\beta_n,  \rho_n \ |\ \Delta) \geq p_{n, \text{max}}$

    In this case, we can denote $\beta_n >0$ from (\ref{kkt_7}), and $p_n = p_{n, \text{max}}$. Substituting $p_n = p_{n, \text{max}}$ into (\ref{change_kkt_2}), we can derive the solution of $\beta_n$  with represented by $\rho_n$ defined as $\widehat{\beta}_n(\rho_n \ |\ \Delta) = \bigg[ \frac{yc_e}{2x_n(\widetilde{r}_{s,n})^3} \frac{\partial \widetilde{r}_{s,n}}{\partial p_n} - \rho_n \frac{\phi_n \mathcal{L} w}{(\widetilde{r}_{s,n})^2}  \frac{\partial \widetilde{r}_{s,n}}{\partial p_n} \bigg]_{(b_n, p_n = b_{n, \text{max}}, p_{n, \text{max}})} \\- 2yc_ex_n(\phi_n \mathcal{L} w)^2 p_n$. 
\end{itemize}
Summarize both cases, and the conclusion could be derived:
\begin{equation}
    \widetilde{p}_n(\rho_n \ |\ \Delta)= \text{min}\{p_{n,\text{max}}, \text{max}\{ 0, \widehat{p}_n(0, \rho_n \ |\ \Delta)\}\}
\end{equation}
\begin{align}
    \widetilde{\beta}(\rho_n \ |\ \Delta)=\left\{
        \begin{array}{ll}
             0, &   \widehat{p}_n (\beta_n,  \rho_n \ |\ \Delta) < p_{n, \text{max}}\\
             \widehat{\beta}_n (\rho_n \ |\ \Delta), &  \text{others}
        \end{array}
    \right. \nonumber
\end{align}
So $\widehat{\alpha}_n (p_n, \rho_n \ |\ \Delta)$ can be represented using (\ref{change_kkt_1}), $\widetilde{\alpha}_n(\rho_n \ |\ \Delta) = 1 + \bigg[\frac{yc_e}{2x_n(\widetilde{r}_{s,n})^3}  \frac{\partial \widetilde{r}_{s,n}}{\partial b_n} + \rho_n \frac{\phi_n \mathcal{L} w}{(\widetilde{r}_{s,n})^2}  \frac{\partial \widetilde{r}_{s,n}}{\partial b_n} \bigg]_{(b_n = b_{n, \text{max}}, p_n = \widetilde{p}_n(\rho_n \ |\ \Delta))}$.

After \textbf{Step 1} and \textbf{Step 2}, we derive the solutions of all the optimal variables except $T$ related to the Lagrange multipliers $\boldsymbol{\rho}$, and we continue to handle the last constraint.

\textbf{Step 3.} With constraint and our previous results, we can denote the updated constraints represented by $(T, \boldsymbol{\rho})$ are as follows
\begin{align}
     \text{From (\ref{cons_new_time}):}~& t_n (b_{n, \text{max}}, \widetilde{p}_n(\rho_n \ |\ \Delta),  \phi^*_n,  \widetilde{f}_n^{\text{DO}}(\rho_n \ |\ \Delta), \notag\\
    &\widetilde{f}_n^{\text{MO}}(\rho_n \ |\ \Delta)) \leq T \label{step_3_1}\\
     \text{From (\ref{kkt_10}):}~ &\rho_n \cdot (t_n (b_{n, \text{max}}, \widetilde{p}_n(\rho_n \ |\ \Delta),  \phi^*_n,  \widetilde{f}_n^{\text{DO}}(\rho_n \ |\ \Delta), \notag \\ &\widetilde{f}_n^{\text{MO}}(\rho_n \ |\ \Delta)) - T) = 0.\label{step_3_2}
\end{align}
We can have the following discussion
\begin{itemize}
    \item \textbf{Case 1.} If $\rho_n = 0$ satisfies constraint (\ref{step_3_1})
    
    In this case, we can simply set $\rho_n = 0$ and we can derive the value of $p_n$, $f_n^{\text{DO}}$ and $f_n^{\text{MO}}$ (i.e., all the optimization variables except $T$), a lower bound of $T$ can be derived since $t_n (b_{n, \text{max}}, \widetilde{p}_n(\rho_n \ |\ \Delta),  \phi^*_n,  \widetilde{f}_n^{\text{DO}}(\rho_n \ |\ \Delta),
    \widetilde{f}_n^{\text{MO}}(\rho_n \ |\ \Delta)) < T$.
    \item \textbf{Case 2.} If $\rho_n = 0$ violates constraint (\ref{step_3_1})

    In this case, $\rho_n >0$ and $t_n (b_{n, \text{max}}, \widetilde{p}_n(\rho_n \ |\ \Delta),  \phi^*_n,  \widetilde{f}_n^{\text{DO}}(\rho_n \ |\ \Delta), \\
    \widetilde{f}_n^{\text{MO}}(\rho_n \ |\ \Delta))-T=0$.

\end{itemize}
Summarize both cases, and the conclusion could be derived
\begin{align} 
    \widetilde{\rho}_n(T \ |\ \Delta)=\left\{
        \begin{array}{ll}
             0, &   t_n (b_{n, \text{max}}, \widetilde{p}_n(\rho_n \ |\ \Delta),  \phi^*_n, \\ &\widetilde{f}_n^{\text{DO}}(\rho_n \ |\ \Delta), \widetilde{f}_n^{\text{MO}}(\rho_n \ |\ \Delta)) < T\\
             \widehat{\rho}_n(T \ |\ \Delta), &  \text{others} 
        \end{array}
    \right. \nonumber
\end{align}
Since parameter $T$ will be the last parameter to derive, which means we already obtained the value of $T$ when we discuss according to the different situation of $\boldsymbol{\rho}$. In the above summarization, the unknown parameter $\boldsymbol{\rho}$ occurs on both sides of the equation, we can propose a new function $R_n(\rho_n \ |\ T, \Delta)$ which satisfies
\begin{align}
    & R_n(\rho_n \ |\ T, \Delta) = -\rho_n, \\
    & \text{if} \ t_n (b_{n, \text{max}}, \widetilde{p}_n(\rho_n \ |\ \Delta),  \phi^*_n, \widetilde{f}_n^{\text{DO}}(\rho_n \ |\ \Delta), \widetilde{f}_n^{\text{MO}}(\rho_n \ |\ \Delta)) < T \notag
\end{align}
\begin{align}
    R_n(\rho_n \ |\ T, \Delta) &= t_n (b_{n, \text{max}}, \widetilde{p}_n(\rho_n \ |\ \Delta),  \phi^*_n, \widetilde{f}_n^{\text{DO}}(\rho_n \ |\ \Delta), \\ &\widetilde{f}_n^{\text{MO}}(\rho_n \ |\ \Delta)) -T,~ \text{if} \ \text{others}. \notag
\end{align}
Since~\cite{zhao2023human} proves that $R_n(\rho_n \ |\ T, \Delta)$ is non-increasing with the increment of $\rho_n$. We can then utilize the multivariate bisection algorithm proposed in~\cite{galvan2017multivariate} and~\cite{zhao2023human} to derive the optimal $\boldsymbol{\rho}$ and $T$.

The complete resource allocation algorithm in Algorithm~\ref{alg_2}.
\begin{algorithm} 
\caption{Resource Allocation Algorithm}
\label{alg_2}
Initialize $ sol^{(0)} = (\boldsymbol{b}^{(0)}, \boldsymbol{p}^{(0)}, \boldsymbol{\phi}^{(0)}, \boldsymbol{(f^{\text{DO}}})^{(0)}, \boldsymbol{(f^{\text{MO}}})^{(0)}, T^{(0)})$, iteration number $j=0$.\\
Calculate $\boldsymbol{x}^{(0)} = \sqrt{\frac{\mathcal{P}_n(\boldsymbol{z})^{(0)}}{\mathcal{M}_n(\boldsymbol{z})^{(0)}}}$, $y^{(0)} = \frac{U(\boldsymbol{z}^{(0)})}{S(\boldsymbol{z}^{(0)})}$, $\boldsymbol{\chi}^{e^{-}(0)}$\\
\Repeat{Convergence or the number of iterations achieves maximum $J$}{
Solve Subproblem 1. Obtain $\boldsymbol{\phi}^{(j+1)}$.\\
Solve Subproblem 2 and obtain $ (\boldsymbol{b}^{(j+1)}, \boldsymbol{p}^{(j+1)},\boldsymbol{(f^{\text{DO}}})^{(j+1)}, \boldsymbol{(f^{\text{MO}}})^{(j+1)}, T^{(j+1)})$.\\
$ sol^{(j+1)} = (\boldsymbol{p}^{(j+1)}, \boldsymbol{B}^{(j+1)}, \boldsymbol{f}^{(j+1)})$.\\
Calculate $\boldsymbol{\chi}^{e^{+} (j+1)}$ and obtain $\boldsymbol{\chi}^{e (j+1)}$. \\
Update $\bold{x}^{(j+1)}$, $y^{(j+1)}$ \\
Set $j \leftarrow j+1$.
}
\end{algorithm}

\section{Experimental results} \label{sec_exp}

In this section, we report the experimental results in detail.

\subsection{Parameter Setting}
In configuring the parameters for our communication system, we use specific settings based on established models and empirical data. The path loss between the large model server and each downstream user is quantified by the equation $128.1 + 37.6 \log(\text{distance})$ dB, where the distance is measured in kilometers. This model also incorporates a standard deviation of 8 dB for shadow fading. The power spectral density of Gaussian noise is set at $-174$ dBm/Hz. Focusing on computational capacities, the maximum GPU computation frequency for each user, $f_{n, \text{max}}^{\text{DO}}$, is set at 7 GHz (utilizing four NVIDIA GeForce RTX 3060 units). For the server, the maximum GPU computation frequency, $f_{\text{max}}^{\text{MO}}$ is set as 100 GHz. The effective switched capacitance parameters, $k_n$ and $k_m$, are both fixed at $10^{-27}$. As for the transmission capabilities, the maximum transmit power of mobile users, $p_{n, \text{max}}$, is 0.2W, while the extraction rates $\phi_n$ vary within the range of $[0, 1]$. The total parameter size considered in our model is $14M$. Lastly, the sum of the cost coefficients $c_e + c_t=1$. 

\subsection{Performance Comparison with baselines}
We consider three different baseline methods in the experiment compared with our proposed method, \textbf{1. Average allocation.} Set each $b_n = b_{n, \text{max}}, p_n = p_{n, \text{max}}, f_n^{\text{DO}} = f_{n, \text{max}}^{\text{DO}}, f_n^{\text{MO}} = \frac{f_{\text{max}}^{\text{MO}}}{N}$ and $\phi_n = 0.5$. \textbf{2. Optimize $\boldsymbol{b}, \boldsymbol{p}, \boldsymbol{s}$ only.} Set each $f_n^{\text{DO}} = f_{n, \text{max}}^{\text{DO}}, f_n^{\text{MO}} = \frac{f_{\text{max}}^{\text{MO}}}{N}$. \textbf{3. Optimize $\boldsymbol{f^{\text{DO}}}, \boldsymbol{f^{\text{MO}}}$ only.} Set each $b_n = b_{n, \text{max}}, p_n = p_{n, \text{max}}$ and $\phi_n = 0.5$.
\begin{figure*}[ht]
    \centering
    \begin{subfigure}{0.33\textwidth}
        \includegraphics[width=0.9\linewidth]{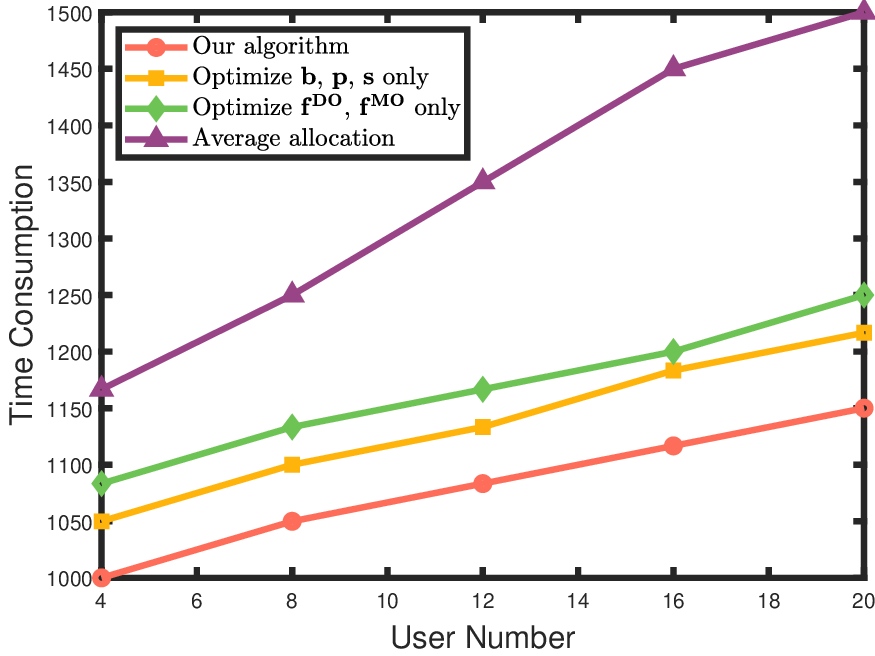}
        \caption{Time consumption versus user number}
        \label{fig_user_time}
    \end{subfigure}%
    \begin{subfigure}{0.33\textwidth}
        \includegraphics[width=0.9\linewidth]{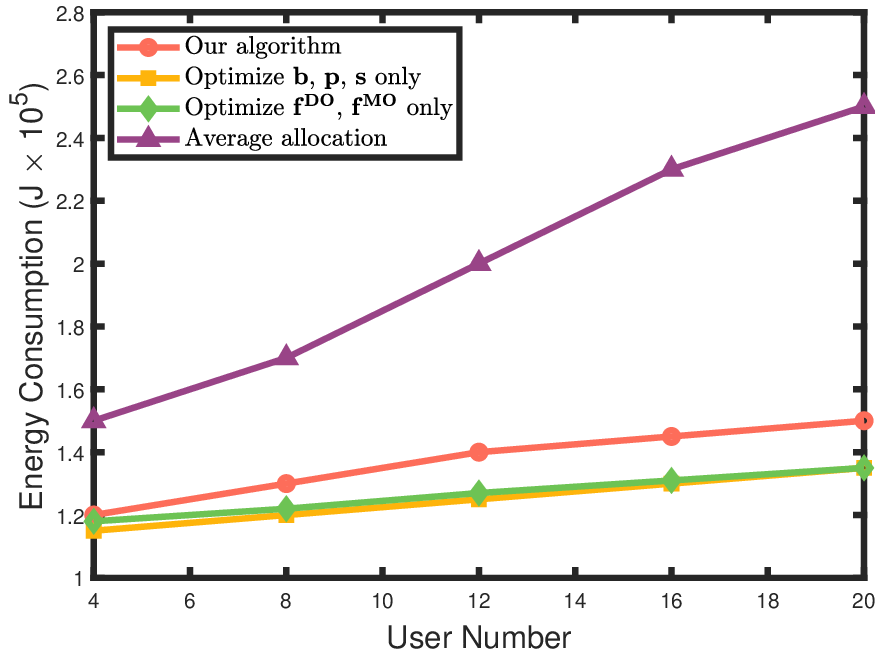}
        \caption{Energy consumption versus user number}
        \label{fig_user_ener}
     \end{subfigure}%
    \begin{subfigure}{0.33\textwidth}
        \includegraphics[width=0.9\linewidth]{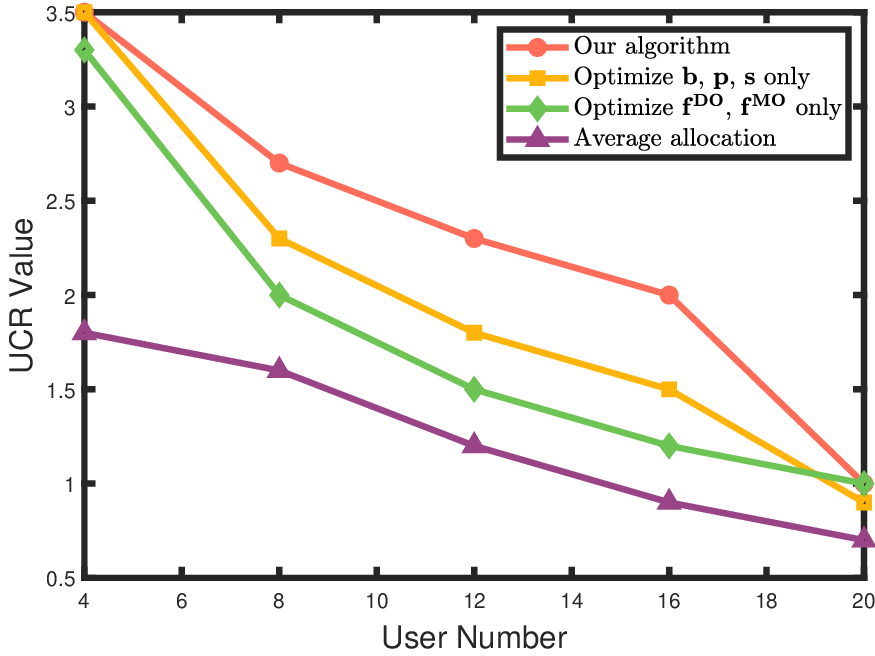}
        \caption{UCR versus user number}
        \label{fig_user_UCR}
    \end{subfigure}%
    \caption{Metrics concerning the number of users.}
    \label{fig_user}
\end{figure*}

\begin{figure*}[ht]
    \centering
    \begin{subfigure}{0.33\textwidth}
        \includegraphics[width=0.9\linewidth]{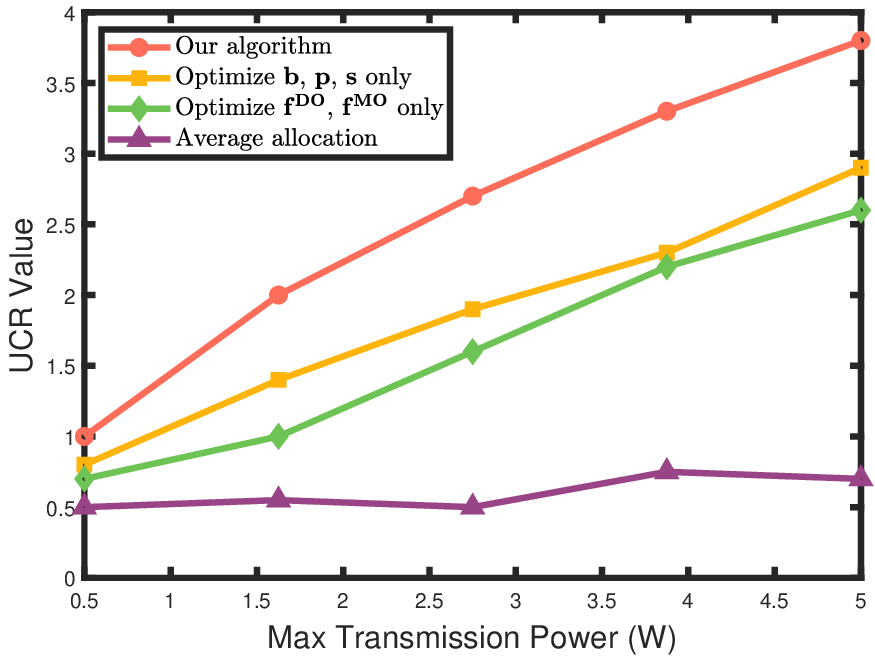}
        \caption{UCR versus max transmission power}
        \label{fig_ener_UCR}
    \end{subfigure}%
    \begin{subfigure}{0.33\textwidth}
        \includegraphics[width=0.9\linewidth]{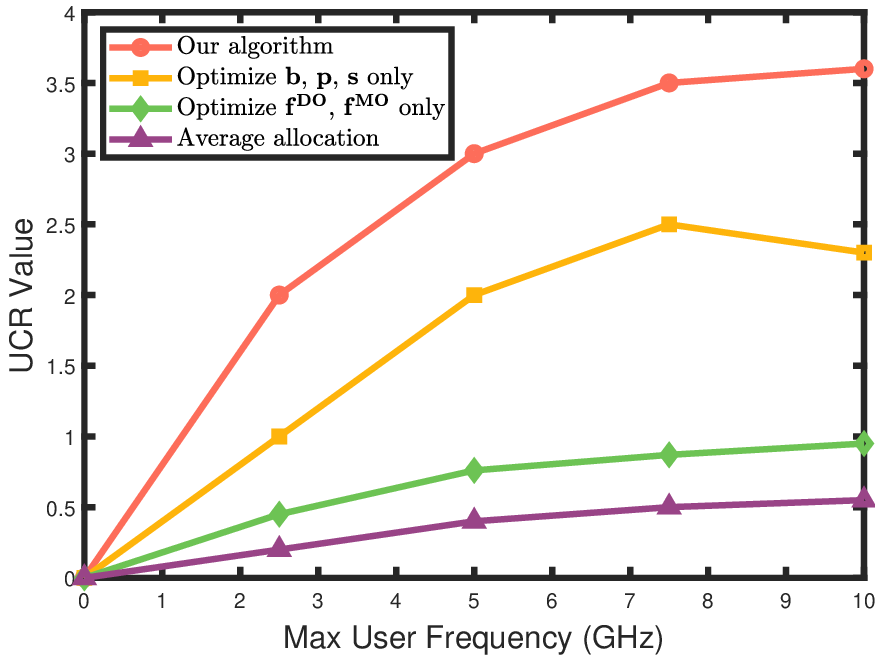}
        \caption{UCR versus max user frequency}
        \label{fig_userfreq_UCR}
     \end{subfigure}%
    \begin{subfigure}{0.33\textwidth}
        \includegraphics[width=0.9\linewidth]{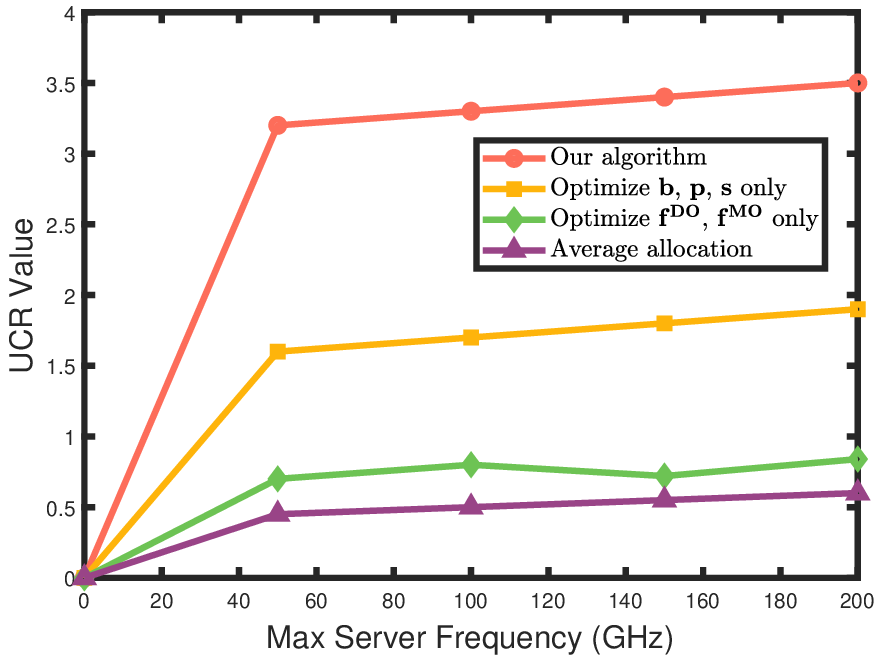}
        \caption{UCR versus max server frequency}
        \label{fig_serfreq_UCR}
    \end{subfigure}%
    \caption{The system utility-cost ratio (UCR) versus various parameters.}
    \label{fig_other_para}
\end{figure*}

\begin{figure}
\centering
\begin{subfigure}{.24\textwidth}
  \centering
  \includegraphics[width=0.9\linewidth]{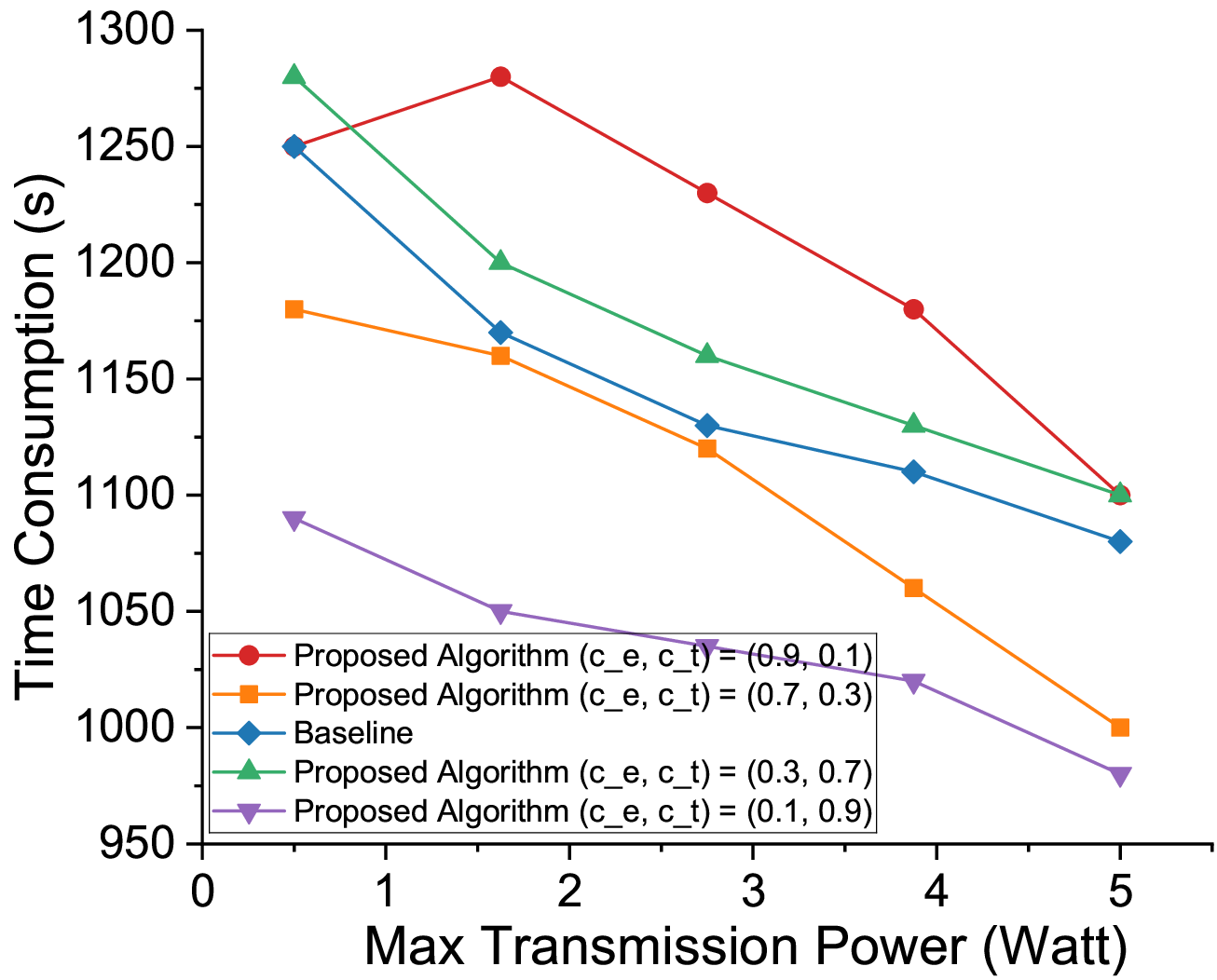}
    \caption{Total time consumption}
    \label{fig_cect_time}
\end{subfigure} \hspace{-15pt}
\begin{subfigure}{.24\textwidth}
  \centering
  \includegraphics[width=0.9\linewidth]{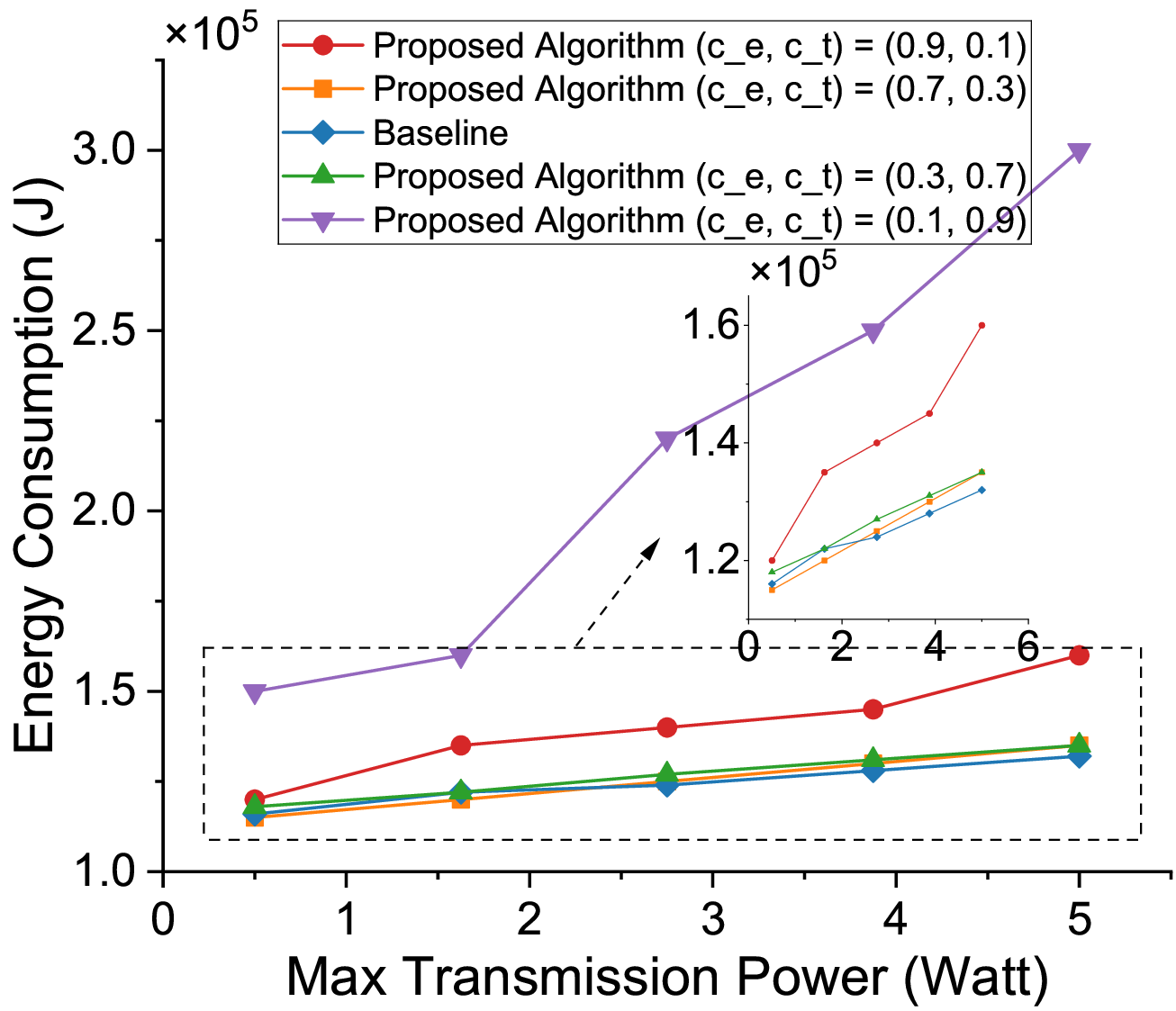}
    \caption{Total energy consumption}
    \label{fig_cect_ener}
\end{subfigure}\vspace{-3pt}
\caption{Experiments with different weight parameters.\vspace{-10pt}}
\label{fig:test}
\end{figure}

\textbf{Impact of user number on UCR.} In Fig.~\ref{fig_user}, we observe total energy consumption increase as the number of users increases. The difference in energy consumption between the optimal (our algorithm) and the least effective (average allocation) performance is approximately 44\%, which demonstrates the advantage of our algorithm compared to the average allocation method. However, allocating fewer resources to each user might reduce the efficiency of the fine-tuning process, reflecting a constraint in resources as the user base expands. As for time consumption, the time efficiency can also be impacted as the number of users increases. Our result shows that the difference in time consumption between our algorithm and the average allocation method is approximately 38.18\%. This implies that our algorithm is more time-efficient than the average allocation method, although an increase in users might lead to longer times the server requires to fine-tune each user's model. 

Furthermore, a decreasing trend can be shown in the average UCR value as the number of users increases, further corroborating the negative impact of resource dilution on user utility. With more users sharing the server's computational resources, each user experiences a reduction in the effectiveness of the offsite-tuning process. The difference in UCR values between our algorithm and the average allocation method is about -42.86\%, indicating that our algorithm is more effective in maintaining higher UCR values, although the specific implications of the UCR values and their impact on the overall system performance must be considered.

The results indicate that our algorithm outperforms the average allocation method regarding energy consumption, UCR value, and time efficiency. As the number of users increases, the impact of resource allocation efficiency on each user becomes more pronounced, potentially leading to a decrease in the efficiency of the fine-tuning process. Our algorithm effectively mitigates these challenges, demonstrating higher energy and time efficiency and an advantage in maintaining UCR values.

\textbf{UCR versus transmission power.} In Fig.~\ref{fig_ener_UCR}, the UCR increases as the transmission power grows for all algorithms since a higher transmission power expands the search space for optimization, allowing for more efficient use of the available power. At lower transmission powers (e.g., 0.5W and 1W), our algorithm has a slightly higher UCR than the others, indicating its effectiveness even in low-power scenarios. However, as the transmission power increases, the performance of our algorithm significantly surpasses the others. This is evident in the UCR values where our algorithm reaches 3.8 at 5W, while others are below this mark, and other algorithms always hold more significant fluctuations and generally lower UCR values compared to ours. For instance, average allocation shows the least improvement and remains below 1 UCR even at higher power levels.

\textbf{UCR versus computation resource.} Fig.~\ref{fig_userfreq_UCR} and Fig.~\ref{fig_serfreq_UCR} shows a clear trend where the UCR increases as the GPU frequency grows. This suggests that higher GPU frequencies provide a wider search space for optimization algorithms, leading to better performance in terms of UCR.
In both two figures, our algorithm consistently outperforms other algorithms. For example, at 8GHz user GPU frequency and 160 GHz server GPU frequency, the difference in UCR compared to the average allocation is 600\% and 700\%, respectively.\vspace{-10pt}

\subsection{Performance When Adapting Weight Parameters}
In our model, $c_e$ and $c_t$, are employed to modulate the optimization's focus. A higher value of $c_e$ compared to $c_t$ signifies an emphasis on minimizing energy consumption within the optimization framework (resp, indicates a prioritization of time efficiency over energy efficiency in the optimization process). To elucidate the impact of these weight parameters, we perform a series of experiments with varied $c_e$ and $c_t$ pairings, assessing our approach against a standard baseline: uniform allocation. This baseline approach allocates resources uniformly among users, setting the weight parameters at an equal balance of $c_e = 0.5, c_t= 0.5$. 

Figures~\ref{fig_cect_time} and~\ref{fig_cect_ener} present the outcomes of total energy and time consumption, respectively, evaluated across five distinct $(c_e, c_t)$ pairings under varying maximum transmission power thresholds. The data indicates as $c_e$ increases and $c_t$ diminishes, there is a noticeable reduction in total energy usage while total time expenditure escalates. This pattern emerges because an augmented $c_e$ (resp, $c_t$) shifts the focus of our optimization approach towards minimizing energy expenses (resp, correspondingly, time consumption). 

\section{Conclusion} \label{sec_con}
In conclusion, our work presents an impactful approach to the offsite fine-tuning of large models in the mobile edge computing environment under the physical layer security and addresses both the privacy concerns and resource allocation problems inherent in the conventional fine-tuning process. We formulated the problem of maximizing the utility-consumption ratio, which balances between maximizing user utility and minimizing system costs. Our proposed optimization algorithm blends the Dinkelbach algorithm, Successive Convex Approximation, fractional programming and alternating optimization techniques, and the algorithm converges effectively. Additionally, the simulated evaluations of our algorithm underscore its superiority over existing methods.

\appendix


\section{Proof of Assumption~\ref{assum_utility_def}} \label{appen_assum_utility_def}
In Assumption~\ref{assum_utility_def}, the concavity of $U_n(x)$, represented by a non-positive second derivative ($U^{''}_n(x) \leq 0 $ {for} $ x > 0$, is indicative of diminishing marginal returns. A concave utility function ensures that any local maximum is also a global maximum, simplifying the search for optimal solutions. Also, for $x>0$, both $U^{''}_n(x)$ and $U^{'}_n(x)$ are well-defined and continuous, which means $U_n(x)$ being twice differentiable. For $U_n(x)$, we do not define a value at $x=0$, if $U_n(x)$ approaches a finite limit as $x \rightarrow 0^{+}$, this limit can be used to define the function at $x=0$. If no such limit exists, $U_n(x)$ remains undefined at $x=0$. The non-decreasing aspect of $U_n(x)$ ensures that the utility increases, or at the very least remains constant, as $x$ (the secrecy rate) increases.

\section{Proof of Theorem~\ref{theo_sub_1_solu}} \label{appen_sub_1_solu}
After applying the KKT conditions to Problem $\mathbb{P}_3$, we have
\begin{align}
   &\frac{\partial \mathcal{L}_1(\boldsymbol{\phi}, \boldsymbol{\eta})}{\partial \phi_n} =  -\frac{1}{\phi_n} + yc_e\big(k_nC_1C_2(f^{\text{DO}}_n)^2\phi_n^{(C_2-1)} + \frac{p_n\mathcal{L}w}{r_{s,n}} \label{lag_1_partial} \\
   &- k_mC_3C_4(f^{\text{MO}}_n)^2\phi_n^{(-C_4-1)}\big) \notag\\
    & + \eta_n\big( \frac{C_1C_2\phi_n^{(C_2-1)}}{f^{\text{DO}}_n} + \frac{\mathcal{L} w}{r_{s,n}} - \frac{C_3C_4\phi_n^{(-C_4-1)}}{f^{\text{MO}}_n} -T\big) = 0.  \notag
\end{align}
And
\begin{equation}   
\eta_n \cdot  (\text{max} \{\frac{C_1 \phi_n ^{C_2}}{f^{\text{DO}}_n} + \frac{\phi_n \mathcal{L} w}{r_{s,n}} + \frac{C_3 \phi_n ^{-C_4}}{f^{\text{MO}}_n} \} - T) = 0 
\end{equation}



From (\ref{lag_1_partial}), we can denote the solution by $\widetilde{\phi}_n(\eta_n | \ y, \boldsymbol{b},  \boldsymbol{p}, \boldsymbol{f^{\text{DO}}}, \boldsymbol{f^{\text{MO}}}, T)$. Note that the right-hand side of~(\ref{lag_1_partial}) is monotonically increasing to $\phi_n$, the optimal solution $\widetilde{\phi}_n(\eta_n \ | \ y, \boldsymbol{b},  \boldsymbol{p}, \boldsymbol{f^{\text{DO}}}, \boldsymbol{f^{\text{MO}}}, T)$ can be obtained using bisection method. Considering~(\ref{cons_1_5_sele}), the Lagrange multiplier should meet the KKT conditions.

\end{document}